\documentclass[reqno]{amsart}
\usepackage[margin=1in]{geometry}

\usepackage{listings}

\usepackage[english]{babel}
\usepackage[utf8]{inputenc}
\usepackage{amsmath}
\usepackage{amssymb}
\usepackage{amsfonts}
\usepackage{amsthm}
\usepackage{graphicx}
\usepackage[colorinlistoftodos]{todonotes}
\usepackage{bm}
\usepackage{url}

\usepackage{multirow}
\usepackage[ruled]{algorithm2e}

\usepackage{natbib}
 \bibpunct[, ]{(}{)}{,}{a}{,}{,}%
\newtheorem{theorem}{Theorem}
\theoremstyle{plain}

\newtheorem{result}{Result}

\numberwithin{equation}{section}
\begin{document}
\title[GPU Algorithms for Large QAPs]{RLT2-based Parallel Algorithms for Solving Large Quadratic Assignment Problems on Graphics Processing Unit Clusters}
\author[Date and Nagi]{Ketan Date and Rakesh Nagi}
\address{Department of Industrial and Enterprise Systems Engineering\newline \indent University of Illinois at Urbana-Champaign, \newline \indent 117 Transportation Building,  Urbana, IL 61801.}%
\email{date2@illinois.edu; nagi@illinois.edu}%


\begin{abstract}
This paper discusses efficient parallel algorithms for obtaining strong lower bounds and exact solutions for large instances of the Quadratic Assignment Problem (QAP). Our parallel architecture is comprised of both multi-core processors and Compute Unified Device Architecture (CUDA) enabled NVIDIA Graphics Processing Units (GPUs) on the Blue Waters Supercomputing Facility at the University of Illinois at Urbana-Champaign. We propose novel parallelization of the Lagrangian Dual Ascent algorithm on the GPUs, which is used for solving a QAP formulation based on Level-2 Refactorization Linearization Technique (RLT2). The Linear Assignment sub-problems (LAPs) in this procedure are solved using our accelerated Hungarian algorithm [Date, Ketan, Rakesh Nagi. 2016. GPU-accelerated Hungarian algorithms for the Linear Assignment Problem. Parallel Computing 57 52-72]. We embed this accelerated dual ascent algorithm in a parallel branch-and-bound scheme and conduct extensive computational experiments on single and multiple GPUs, using problem instances with up to 42 facilities from the QAPLIB. 
The experiments suggest that our  GPU-based  approach is scalable and it can be used to obtain  tight lower bounds on large QAP instances. Our accelerated branch-and-bound scheme is able to comfortably solve Nugent and Taillard instances (up to 30 facilities) from the QAPLIB, using modest number of GPUs. 
\end{abstract}

\maketitle


\section{Introduction}\label{sec:intro}

The Quadratic Assignment Problem (QAP) is one of the oldest mathematical problems in the literature and it has received substantial attention from the researchers around the world. QAP was originally introduced by \citet{koopmans1957} as a mathematical model to locate indivisible economical activities (such as facilities) on a set of locations and the cost of the assignment is a function of both distance and flow. The objective is to assign each facility to a location so as to minimize a quadratic cost function. The generalized mathematical formulation for the QAP, given by \citet{lawler1963}, can be written as follows: 
\begin{align}
\text{QAP: }\min \quad & \sum_{i=1}^n \sum_{p=1}^n b_{ip}x_{ip} + \sum_{i=1}^n \sum_{j=1}^n \sum_{p=1}^n \sum_{q=1}^n f_{ij}d_{pq}x_{ip}x_{jq}; & \label{eq:qap:obj}\\
\text{s.t.} & \sum_{p=1}^n x_{ip} = 1 \quad & \forall i = 1,\dots, n;  & \label{eq:qap:con1}\\
& \sum_{i=1}^n x_{ip} = 1 \quad & \forall p = 1,\dots, n;  & \label{eq:qap:con2}\\
& x_{ip} \in \left \{ 0, 1 \right \} \quad & \forall i, p = 1, \dots, n. \label{eq:qap:con3} &
\end{align}
The decision variable $x_{ip}=1$, if facility $i$ is assigned to location $p$ and 0 otherwise. Constraints \eqref{eq:qap:con1} and \eqref{eq:qap:con2} enforce that each facility should be assigned to exactly one location and each location should be assigned to exactly one facility. $b_{ip}$ is the fixed cost of assigning facility $i$ to location $p$; $f_{ij}$ is the flow between the pair of facilities $i$ and $j$; and $d_{pq}$ is the distance between the pair of locations $p$ and $q$.

Despite having the same constraint set as the Linear Assignment Problem (LAP), the QAP is a strongly NP-hard problem \citep{sahni1976}, i.e., it cannot be solved efficiently within a guaranteed time limit. Additionally, it is difficult to find a provable $\epsilon$-optimal solution to the QAP.  The quadratic nature of the objective function also adds to the solution complexity. One of the ways of solving the QAP is to convert it into a Mixed Integer Linear Program (MILP) by introducing additional variables and constraints. Different linearizations have been proposed by \citet{lawler1963}, \citet{kaufman1978}, \citet{frieze1983} and \citet{adams1994}. Table \ref{tbl:qaplin}  presents a comparison of these various linearizations in terms of number of variables and constraints. Many formulations and algorithms have been developed over the years for solving the QAP optimally or sub-optimally. For a list of references on the QAP,  readers are directed to the survey papers by \citet{burkard2002} and \citet{loiola2007}.

\begin{table}
\begin{center}
\caption{Linearization models for QAP.}
\label{tbl:qaplin}
{
\begin{tabular}{|l|c|c|c|c|}
\hline
\textbf{Linearization Model} & \textbf{Binary Variables} & \textbf{Continuous Variables} & \textbf{Constraints} \tabularnewline
\hline
\citet{lawler1963} & $O(n^4)$ & -- & $O(n^4)$\tabularnewline
\hline
\citet{kaufman1978} & $O(n^2)$ & $O(n^2)$ & $O(n^2)$\tabularnewline
\hline
\citet{frieze1983} & $O(n^2)$ & $O(n^4)$ & $O(n^4)$\tabularnewline
\hline
\citet{adams1994} RLT1 & $O(n^2)$ & $O(n^4)$ & $O(n^4)$\tabularnewline
\hline
\citet{adams2007} RLT2 & $O(n^2)$  & $O(n^6)$ & $O(n^6)$ \tabularnewline
\hline
\citet{hahn2012} RLT3 & $O(n^2)$ & $O(n^8)$ & $O(n^8)$\tabularnewline
\hline
\end{tabular}
}
\end{center}
\end{table}

The main advantage of formulating the QAP as an MILP is that we can relax the integrality restrictions on the variables and solve the resulting linear program. The objective function value obtained from this LP solution can be used as a lower bound in the exact solution methods such as  branch-and-bound. The most promising formulation was obtained by \citet{adams1994} by applying level-1 refactorization and linearization technique (RLT1) to the QAP. This was considered to be one of the best linearizations at the time, because it yielded strong LP relaxation bound. \citet{adams1994} developed an iterative algorithm based on the Lagrangian dual ascent to obtain a lower bound for the QAP. Later \citet{hahn1998} developed an augmented dual ascent scheme (with simulated annealing), which yielded a lower bound which was close to the LP relaxation bound. 
This linearization technique was extended to RLT2 by \citet{adams2007}, which contains $O(n^6)$ variables and constraints; and RLT3 by \citet{hahn2012}, which contains $O(n^8)$ variables and constraints. These two formulations provide much stronger lower bounds as compared to RLT1, and for many problem instances they are able to match the optimal objective value of the QAP. However, it is extremely difficult to solve these linearization models using primal methods, because of the curse of dimensionality.  \citet{ramakrishnan2002} used Approximate Dual Projective (ADP) method to solve the LP relaxation of the RLT2 formulation of \cite{ramachandran1996}, which was limited to the problems with size $n = 12$. \citet{adams2007} and \citet{hahn2012} developed a dual ascent based algorithms to find strong lower bounds on RLT2 and RLT3 respectively, and used them to solve QAPs with $n\leq 30$.  
As observed by \citet{hahn2012}, LP relaxations of RLT2 and RLT3 provide strong lower bounds on the QAP, with RLT3 being the strongest. However, due to the large number of variables and constraints in RLT3, tremendous amount of memory is required to handle even the small/medium-sized QAP instances. In comparison, the RLT2 formulation has lesser memory requirements and it provides sufficiently strong lower bounds. 

For obtaining a lower bound on the QAP using RLT2 dual ascent, we need to solve $O(n^4)$ LAPs and update $O(n^6)$ Lagrangian multipliers during each iteration, which can become computationally intensive. However, as described in Section \ref{sec:rlt2da}, this algorithm can benefit from parallelization on an appropriate parallel computing architecture. In recent years, there have been significant advancements in the graphics processing hardware. Since graphics processing tasks generally require high data parallelism, the Graphics Processing Units (GPUs) are built as compute-intensive, massively parallel machines, which provide a cost-effective solution for high performance computing applications. Recently \citet{goncalves2017} developed a GPU-based dual ascent algorithm for RLT2, which shows significant parallel speed up as compared to the sequential algorithm. Although it is very promising, their algorithm is limited to a single GPU and not scalable for large problems. 
In our previous work \citep{date2016gpu}, we developed a GPU-accelerated Hungarian algorithm, which was shown to outperform state-of-the-art sequential and multi-threaded CPU implementations. In this work, we are proposing a distributed version of the RLT2 Dual Ascent algorithm (which makes use of our GPU-accelerated LAP solver)  and a parallel branch-and bound algorithm, specifically designed for the CUDA enabled NVIDIA GPUs, for solving large instances of the QAP to optimality. These algorithms make use of the hybrid MPI+CUDA architecture, on the GPU cluster offered by the Blue Waters Supercomputing facility at the University of Illinois at Urbana-Champaign. This research is radical because, to the best of our knowledge, this is the first scalable GPU-based algorithm that can be used for solving large QAPs in a grid setting.  

The rest of the paper is organized as follows. Section \ref{sec:rlt2da} 
describes the RLT2 formulation and the concepts of the sequential dual ascent algorithm. Sections \ref{sec:prlt2da}  and \ref{sec:ch4:bnb} describe the various stages of our parallel algorithm, and an implementation on the multi-GPU architecture. Section \ref{sec:ch4:comp} contains the experimental results on the instances from the QAPLIB \citep{qaplib1997}. Finally, the paper is concluded in Section \ref{sec:qap:concl} with a summary and some directions for future research.


\section{RLT2 Formulation and Dual Ascent Algorithm} \label{sec:rlt2da}

In this section we will explain in detail the concepts of RLT2 formulation and the dual ascent algorithm from \citet{adams1994} and \citet{adams2007}, which can be used to obtain lower bounds on the QAP.

\subsection{RLT2 Formulation for the QAP} \label{sec:rlt2}

As explained by \citet{adams2007}, the refactorization-linearization technique can be applied to the QAP formulation \eqref{eq:qap:obj}-\eqref{eq:qap:con3}, to obtain an instance of MILP. Henceforth, it is assumed that the indices $i$, $j$, $p$, $q$, etc., go from $1$ to $n$ unless otherwise stated. Initially, in the ``refactorization'' step, the constraints \eqref{eq:qap:con1} and  \eqref{eq:qap:con2} are multiplied by variables $x_{ip}, \forall i, p$. After removing the invalid variables of the form $x_{ip}\cdot x_{iq}$ and omitting the trivial constraints $x_{ip} \cdot x_{ip} = x_{ip}$ we obtain $2n^2(n-1)$ new constraints of the form $\sum_{j\neq i} x_{jq} \cdot x_{ip} = x_{ip},\; \forall i, p, q$; and $\sum_{q\neq p} x_{jq} \cdot x_{ip} = x_{ip},\; \forall i, j, p$. Then, in the ``linearization'' step, the product $x_{ip}\cdot x_{jq}$ is replaced by a new variable $y_{ijpq}$ with cost coefficient $C_{ijpq} = f_{ij}\cdot d_{pq}$; and a set of $\frac{n^2(n-1)^2}{2}$ constraints of the form $y_{ijpq} = y_{jiqp}$ are introduced to signify the symmetry of multiplication. The resulting formulation is called RLT1 by \citet{adams2007}, which is depicted below:
\begin{flalign}
\text{RLT1: } \min & \sum_{i} \sum_{p} b_{ip}x_{ip} + \sum_{i} \sum_{j \neq i} \sum_{p} \sum_{q \neq p} C_{ijpq} y_{ijpq}; \label{eq:rlt1:obj}\\
\text{s.t. } & \eqref{eq:qap:con1}-\eqref{eq:qap:con3} \notag \\
             & \sum_{q \neq p} y_{ijpq} = x_{ip}, & \forall (i\neq j ,p); \label{eq:rlt1:con1} \\
             & \sum_{j \neq i} y_{ijpq} = x_{ip}, & \forall (i, p\neq q); \label{eq:rlt1:con2} \\
             & y_{ijpq} = y_{jiqp}, & \forall (i< j, p\neq q); \label{eq:rlt1:con3} \\
             & y_{ijpq} \geq 0, & \forall (i\neq j, p\neq q). \label{eq:rlt1:con4}
\end{flalign}

\begin{result}\label{obs:rlt1equiv}
The RLT1 formulation is equivalent to the QAP, i.e., a feasible solution to RLT1 is also feasible to the QAP with the same objective function value \citep{adams1994}.
\end{result}

Now the refactorization-linearization technique is applied on the RLT1 formulation. During the refactorization step, the constraints \eqref{eq:rlt1:con1}-\eqref{eq:rlt1:con4} are multiplied by variables $x_{ip}, \forall i, p$. The product $y_{jkqr} \cdot x_{ip}$ is replaced by a new variable $z_{ijkpqr}$, with a cost coefficient of $D_{ijkpqr}$. The resulting RLT2 formulation is depicted below:
\begin{flalign}
\text{RLT2: } \min & \sum_{i} \sum_{p} b_{ip}x_{ip} + \sum_{i} \sum_{j \neq i} \sum_{p} \sum_{q \neq p} C_{ijpq}y_{ijpq} \notag \\ 
 & + \sum_{i} \sum_{j \neq i} \sum_{k \neq i,j} \sum_{p} \sum_{q \neq p} \sum_{r \neq p,q}  D_{ijkpqr}z_{ijkpqr}; \label{eq:rlt2:obj}\\
\text{s.t. } & \eqref{eq:qap:con1}-\eqref{eq:qap:con3}; \notag \\
& \eqref{eq:rlt1:con1}-\eqref{eq:rlt1:con4}; \notag \\
             & \sum_{r \neq p, q} z_{ijkpqr} = y_{ijpq}, & \forall (i\neq j \neq k, p\neq q); \label{eq:rlt2:con1} \\
             & \sum_{k \neq i, j} z_{ijkpqr} = y_{ijpq}, & \forall (i\neq j, p\neq q\neq r); & \label{eq:rlt2:con2} \\
             & z_{ijkpqr} = z_{ikjprq} = z_{jikqpr}  = z_{jkiqrp} = z_{kijrpq} = z_{kjirqp}, & \forall (i < j< k, p\neq q \neq r); \label{eq:rlt2:con3} \\
             & z_{ijkpqr} \geq 0, &  \forall (i\neq j\neq k, p\neq q\neq r). \label{eq:rlt2:con4}
\end{flalign}

\begin{result}\label{obs:rlt2equiv}
The RLT2 formulation is equivalent to the QAP, i.e., a feasible solution to RLT2 is also feasible to the QAP with the same objective function value \citep{adams2007}.
\end{result}




The main advantage of using RLT2 formulation is that its LP relaxation (LPRLT2) obtained by relaxing the binary restrictions on $x_{ip}$ yields much stronger lower bounds than any other linearization from the literature. However, since this formulation has a large number of variables and constraints, primal methods are likely to fail for large QAPs (as observed by \citet{ramakrishnan2002}. \citet{adams1994} and  \citet{adams2007} addressed this problem by developing a solution procedure based on Lagrangian dual ascent. In the next sections we will briefly discuss Lagrangian duality and then explain the Lagrangian dual ascent algorithm for RLT2.


\subsection{Lagrangian Duality} \label{sec:ld}

Duality is an important concept in the theory of optimization. The Primal problem (P) and its Dual (D) share a very special relationship, known as the ``weak duality.'' If $\nu(\cdot)$ represents the objective function of a problem, then the weak duality states that $\nu(D) \leq \nu(P)$ for minimization problem.  Many algorithms make use of this relationship, in cases where one of these problems is easier to solve than its counterpart. 
The basis of these constructive dual techniques is the Lagrangian relaxation. Let us consider the following simple optimization problem:
\begin{equation}
\text{P: } \min \quad \mathbf{cx}; \quad \text{s.t.} \quad \mathbf{Ax} \geq \mathbf{b}; \quad \mathbf{x} \in X.
\end{equation}
Here, $\mathbf{Ax} \geq \mathbf{b}$ represent the complicating constraints and $\mathbf{x} \in X$ represent simple constraints. Then we can relax the complicating constraints and add them to the objective function using non-negative Lagrange multipliers $\mathbf{u}$, which gives rise to the following Lagrangian relaxation:
\begin{equation}
\text{LRP($\mathbf{u}$): }  \min \quad \mathbf{cx} + \mathbf{u}(\mathbf{b}-\mathbf{Ax}); \quad \text{s.t.} \quad  \mathbf{x} \in X.
\end{equation}

For any $\mathbf{u} \geq 0$, $\nu(\text{LRP}(\mathbf{u}))$ provides a lower bound on $\nu(P)$, i.e., $\nu(\text{LRP}(\mathbf{u})) \leq \nu(P)$. To find the best possible lower bound, we  solve the  Lagrangian dual problem $\text{LD($\mathbf{u}$): } \max_{\mathbf{u} \geq 0} \nu(\text{LRP}(\mathbf{u}))$.
Hence, the primary goal in these solution procedures is to systematically search for the Lagrange multipliers which maximize the objective function value of the Lagrangian dual problem. The following two solution procedures are most commonly used for obtaining these dual multipliers.


\paragraph{Subgradient Lagrangian Search.}  The subgradient search method operates on two important observations: (1) $\nu(\text{LRP}(\mathbf{u}))$ is a piecewise-linear concave function of $\mathbf{u}$; and (2) At some point $\mathbf{\hat{u}} \geq 0$, if $\mathbf{\hat{x}} \in X$ is a solution to LRP($\mathbf{\hat{u}}$), then $(\mathbf{b}-\mathbf{A\hat{x}})$ represents a valid subgradient of the function $\nu(\text{LRP}(\mathbf{\hat{u}}))$.
The subgradient search procedure is very similar to the standard gradient ascent procedure, where we advance along the (sub)gradients of the objective function until we reach some solution that is no longer improving. At that point, we calculate the new (sub)gradients and continue. The only disadvantage of using subgradients instead of the gradient is that it is difficult to characterize an accurate step-size which is valid for all the active subgradients. Therefore, taking an arbitrary step along the subgradients might worsen the objective function from time to time. However, for specific step-size rules, it is proved that the procedure  converges to the optimal solution asymptotically.


\paragraph{Lagrangian Dual Ascent.} Instead of using the subgradients in a naive fashion, they can be used to precisely figure out both the ascent direction and the step-size that gives us the ``best'' possible improvement in the objective function $\nu(\text{LRP}(\mathbf{u}))$. This is the crux of the dual ascent procedure. During each iteration of the dual ascent procedure, an optimization problem is solved to find a direction $\mathbf{d}$ for some dual solution $\mathbf{\hat{u}}$, which creates a positive inner product with every subgradient of $\nu(\text{LRP}(\mathbf{\hat{u}}))$, i.e., $\mathbf{d}(\mathbf{b} - \mathbf{Ax}) > 0, \forall \mathbf{x} \in X(\mathbf{\hat{u}})$. If no such direction is found, then the solution $\mathbf{\hat{u}}$ and corresponding $\mathbf{\hat{x}}$ is an optimal solution. Otherwise, the ``best'' step-size is established which gives the maximum improvement in the objective function along $\mathbf{d}$ to find a new dual solution. The most difficult part of the dual ascent algorithm is to find the step-size $\lambda$ that will provide a guaranteed ascent, while maintaining the feasibility of all the previous primal solutions, and more often than not, finding the optimal step-size is an NP-hard problem. However, the salient feature of the Lagrangian dual of RLT2 linearization is that the improving direction and step-size can be found without having to solve any optimization problem. This can be achieved by doing simple sensitivity analysis and maintaining the complementary slackness for the nonbasic $\mathbf{x}$, $\mathbf{y}$ and $\mathbf{z}$ variables in the corresponding LAPs. In the next section, we will discuss the features of the RLT2 linearization and its Lagrangian dual. 


\subsection{Sequential RLT2-DA Algorithm}

Let us consider the LP relaxation of the RLT2 formulation. Initially, the constraints \eqref{eq:rlt1:con3} and 
\eqref{eq:rlt2:con3} are relaxed and added to the objective function using the Lagrange multipliers $\mathbf{v} = \langle u_{ijpq}, 
v_{ijkpqr} \rangle$, to obtain the Lagrangian relaxation LRLT2. Let $\alpha, \beta, \gamma, \delta, \xi, \psi$ represent the dual variables corresponding to the constraints \eqref{eq:qap:con1}, \eqref{eq:qap:con2}, \eqref{eq:rlt1:con1}, \eqref{eq:rlt1:con2}, \eqref{eq:rlt2:con1}, \eqref{eq:rlt2:con2} respectively. 
Then for some fixed $\mathbf{\hat{v}}$, the Lagrangian relaxation $\text{LRLT2}(\mathbf{\hat{v}})$ and its dual $\text{DLRLT2}(\mathbf{\hat{v}})$ can be written as follows.
\begin{flalign}
\text{LRLT2($\mathbf{\hat{v}}$): } \min & \sum_{i} \sum_{p} b_{ip}x_{ip} + \sum_{i} \sum_{j \neq i} \sum_{p} \sum_{q \neq p} (C_{ijpq} - \hat{u}_{ijpq} )y_{ijpq} \notag \\ 
 & + \sum_{i} \sum_{j \neq i} \sum_{k \neq i,j} \sum_{p} \sum_{q \neq p} \sum_{r \neq p,q}  (D_{ijkpqr} - \hat{v}_{ijkpqr})z_{ijkpqr}; \label{eq:lrlt2:obj}\\
\text{s.t. } & \eqref{eq:qap:con1}-\eqref{eq:qap:con2}; \notag \\
& \eqref{eq:rlt1:con1}-\eqref{eq:rlt1:con2}; \notag \\
            & \eqref{eq:rlt2:con1}-\eqref{eq:rlt2:con2}; && \notag \\
            & x_{ip} \geq 0; \quad y_{ijpq} \geq 0; \quad z_{ijkpqr} \geq 0. 
\end{flalign}

\begin{flalign}
\text{DLRLT2($\mathbf{\hat{v}}$): } \max & \sum_{i} \alpha_i + \sum_{p} \beta_p; \label{eq:dlrlt2:obj}\\
\text{s.t. } & \alpha_i + \beta_p - \sum_{j\neq i} \gamma_{ijp} - \sum_{q \neq p} \delta_{ipq} \leq b_{ip}, && \forall i, p; \label{eq:dlrlt2:con1} \\
             & \gamma_{ijp} + \delta_{ipq} - \sum_{k \neq i,j} \xi_{ijkpq} - \sum_{r \neq p, q} \psi_{ijpqr} \leq C_{ijpq} - \hat{u}_{ijpq}, && \forall (i\neq j, p\neq q); \label{eq:dlrlt2:con2} \\
              & \xi_{ijkpq} + \psi_{ijpqr} \leq D_{ijkpqr} - \hat{v}_{ijkpqr}, && \forall (i\neq j\neq k, p\neq q\neq r); \label{eq:dlrlt2:con3} \\
             & \alpha_{i}, \beta_{p}, \gamma_{ijp}, \delta_{ipq}, \xi_{ijkpq}, \psi_{ijpqr} \in \mathbb{R}, && \forall (i \neq j \neq k, p \neq q \neq r) && \label{eq:dlrlt2:con4}
\end{flalign}


\paragraph{LAP Solution.} The problem $\text{DLRLT2}(\mathbf{\hat{v}})$ can be solved using the decomposition principle explained by \citet{adams2007}. To maximize the dual objective function \eqref{eq:dlrlt2:obj} with respect to constraints \eqref{eq:dlrlt2:con1}, we need large values of $\alpha$ and $\beta$, for which the term $\sum_{j\neq i} \gamma_{ijp} + \sum_{q \neq p} \delta_{ipq}$ needs to be maximized subject to constraints \eqref{eq:dlrlt2:con2}. This requires large values of $\gamma$ and $\delta$, for which, the term $\sum_{k \neq i,j} \xi_{ijkpq} + \sum_{r \neq p, q} \psi_{ijpqr}$ needs to be maximized with respect to constraints \eqref{eq:dlrlt2:con3}. Thus we have a three stage problem, as seen in Fig.\ \ref{fig:3st-rlt2}. 

\begin{figure}
\begin{center}
\includegraphics[width=0.75\columnwidth, page=4]{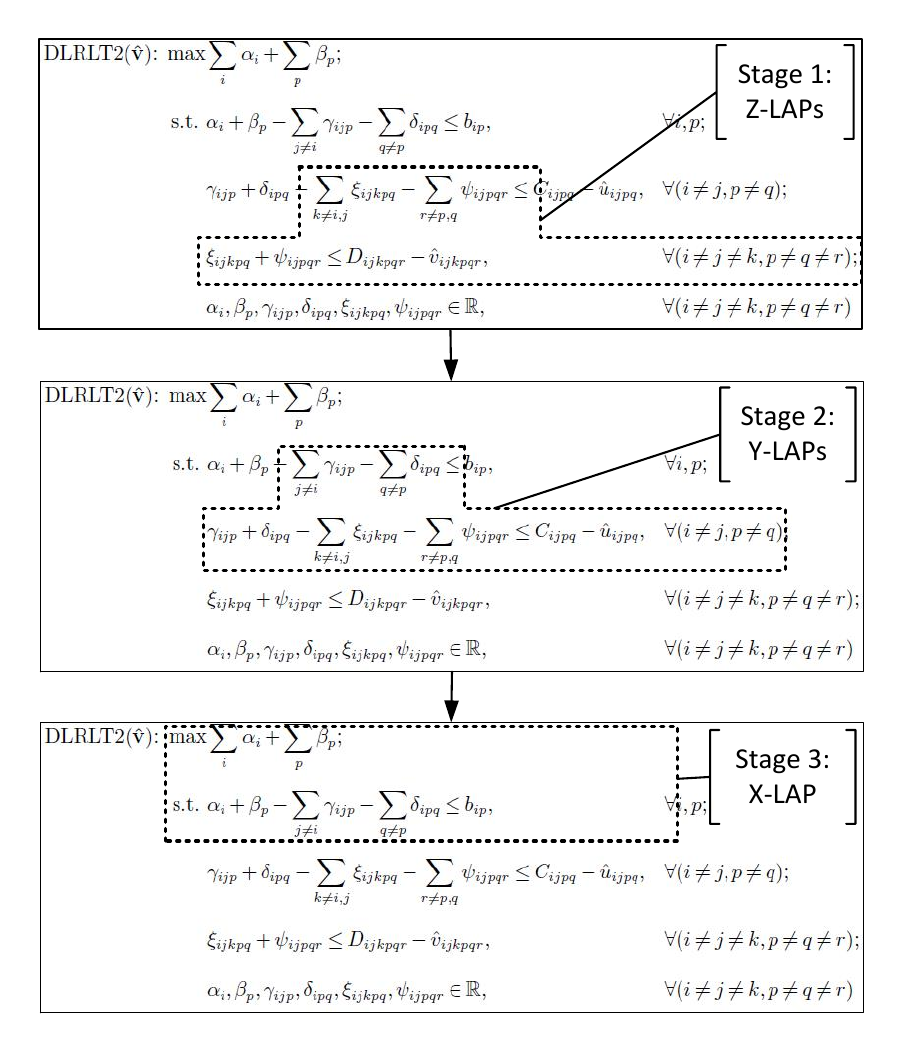}
\caption{Three stage solution of LRLT2($\mathbf{\hat{v}}$).}
\label{fig:3st-rlt2}
\end{center}
\end{figure}

In the first stage, for each $(i, j, p, q)$, with $i \neq j$ and $p \neq q$, we need to solve the problems: 
\begin{flalign}
\Theta_{ijpq}(\mathbf{\hat{v}}) = \max \left \lbrace \sum_{k \neq i,j} \xi_{ijkpq} + \sum_{r \neq p, q}  \psi_{ijpqr} \middle | \xi_{ijkpq} + \psi_{ijpqr} \leq \hat{D}_{ijkpqr},  \forall (k \neq i, j; r \neq p, q) \right\rbrace,&&
\end{flalign}
which are nothing but $n^2(n-1)^2$ Z-LAPs in their dual form (with modified cost coefficients). In the second stage, for each $(i, p)$, we need to solve the problems:
\begin{flalign}
\Delta_{ip}(\mathbf{\hat{v}}) = \max \left \lbrace \sum_{j \neq i} \gamma_{ijp} + \sum_{q \neq p}  \delta_{ipq} \middle | \gamma_{ijp} + \delta_{ipq} \leq \hat{C}_{ijpq} + \Theta_{ijpq}(\mathbf{\hat{v}}), \forall (j \neq i; q \neq p) \right\rbrace,&&
\end{flalign}
which are nothing but $n^2$ Y-LAPs in their dual form (with modified cost coefficients). In the final stage, we need to solve a single X-LAP (with modified cost coefficients):
\begin{flalign}
\nu(\text{LRLT2($\mathbf{\hat{v}}$)}) = \nu(\text{DLRLT2($\mathbf{\hat{v}}$)}) = \max \left \lbrace \sum_{i} \alpha_{i} + \sum_{p}  \beta_{p} \middle | \alpha_{i} + \beta_{p} \leq b_{ip} + \Delta_{ip}(\mathbf{\hat{v}}), \forall (i, p) \right\rbrace,&&
\end{flalign}
which gives us the required lower bound on RLT2. We can see that there are $O(n^4)$ LAPs and the number of cost coefficients in each LAP is $O(n^2)$. The worst case complexity of any primal-dual LAP algorithm for an input matrix with $n^2$ cost coefficients, is $O(n^3)$. Therefore, the overall solution complexity for solving DLRLT2($\mathbf{\hat{v}}$) is $O(n^7)$.

\paragraph{Dual Ascent.} The Lagrangian dual problem for LRLT2 is to find the best set of multipliers $\mathbf{v}^*$, so as to maximize the objective function value $\nu(\text{LRLT2})$, i.e.,
\begin{equation}\label{eq:ldrlt2}
\text{LDRLT2: } \max_{\mathbf{v}} \left\lbrace\nu(\text{LRLT2}(\mathbf{v})) \right\rbrace.
\end{equation} 
Since LRLT2($\mathbf{v}$) exhibits integrality property, due to the theorem by \citet{geoffrion1974}, the objective function value $\nu(\text{LDRLT2})$ cannot exceed the linear programming relaxation bound $\nu(\text{LPRLT2})$, obtained by relaxing the binary restrictions \eqref{eq:qap:con3}. Therefore, we can assert the following inequality:
\begin{equation}
\label{eq:lgduality}
\nu(\text{LDRLT2}) \leq \nu(\text{LPRLT2}) \leq \nu(\text{RLT2}) = \nu(\text{QAP}).
\end{equation}

To solve LDRLT2, one could employ a standard dual ascent algorithm. However, for LDRLT2, finding an ascent direction and a step-size can be done relatively easily, without having to solve any optimization problem. To this end, we will now describe the principle behind the Lagrangian dual ascent for LDRLT2. 

\begin{enumerate}
\item Let $\pi(\cdot)$ denote the reduced cost (or dual slack) of an LAP variable. Then, for some variable $z_{ijkpqr}$ in an optimal LAP solution,
\begin{equation}\label{eq:compslack}
z_{ijkpqr} = 1 \implies \pi(z_{ijkpqr}) = 0; \text{ and } z_{ijkpqr} = 0 \implies \pi(z_{ijkpqr}) \geq 0.
\end{equation}

\item From Equation \eqref{eq:rlt2:con3}, we know that for any $(i, j, k, p, q, r) : i < j < k, p \neq q \neq r$, the variable $z_{ijkpqr}$ is one of the six ``symmetrical'' variables appearing in that particular constraint, and in an optimal QAP solution, the values of all the six variables should be the same. 

\item If some $z_{ijkpqr} = 0$ and one of its symmetrical variables $z_{jikqpr} = 1$, then for the constraint $z_{ijkpqr} = z_{jikqpr}$, the direction $(1, -1)$ provides a natural direction of ascent for $(\hat{v}_{ijkpqr}, \hat{v}_{jikqpr})$, because it is a valid subgradient of LRLT2. To obtain a new dual solution, a step may be taken along this direction, i.e.,  $\hat{v}_{ijkpqr}$ may be increased (i.e., $D_{ijkpqr}$ may be decreased) and  $\hat{v}_{jikqpr}$ may be decreased (i.e., $D_{jikqpr}$ may be increased), using a valid step-size. 

\item While determining the step-size, the most important criterion is that the feasibility of the current dual variables $\alpha, \beta, \gamma, \delta, \xi, \psi$ must be maintained. According the constraint \eqref{eq:dlrlt2:con3}, infeasibility is incurred in the dual space if the new $\pi'(z_{ijkpqr}) = D_{ijkpqr} - \hat{v}_{ijkpqr} - \xi_{ijkpq} - \psi_{ijpqr} < 0$. This means that  $D_{ijkpqr}$ is  allowed to decrease by at most $\pi(z_{ijkpqr})$, and consequently, the symmetrical cost coefficient $D_{jikqpr}$ can be increased by the same amount. Since $z_{jikqpr}$ is basic, there is a good chance that this adjustment will increase $\nu(\text{LDRLT2})$ by some non-negative value, and therefore, this is a ``strong'' direction of ascent. 

\item For some other pair of variables, if $z_{ijkpqr} = z_{ikjprq} = 0$, then the direction $(1, -1)$ is also a valid direction, i.e., the cost coefficient $D_{ijkpqr}$ can be decreased by at most $\pi(z_{ijkpqr})$ and $D_{ikjprq}$ can be increased by the same amount, without incurring any dual infeasibility. Since both variables are non-basic, there will be no change in $\nu(\text{LDRLT2})$. This direction is a ``weak'' direction of ascent.

\item In an ``optimal'' dual ascent scheme, we would need to find ascent directions which will be ``strong'' for every pairwise constraint from Equation \eqref{eq:rlt2:con3}, and finding such directions would require significant computational effort. However, we can easily find a direction that is ``strong'' only for a subset of pairwise constraints, which may provide a non-negative increase in $\nu(\text{LDRLT2})$. In other words, we can select a non-basic variable $z_{ijkpqr}$, decrease its cost coefficient by some amount $0 < \epsilon \leq \pi(z_{ijkpqr})$ and increase the cost coefficients of the five symmetrical variables by some fraction of $\epsilon$. If some of the directions happen to be ``strong,'' then the objective function $\nu(\text{LDRLT2})$ will experience non-negative increase, otherwise it will remain the same. This is the crux of the dual ascent procedure. Mathematically, we adjust the  dual multipliers using the rule:
\begin{align}
\label{eq:da:rule1}
v_{ijkpqr} \leftarrow v_{ijkpqr} + \kappa^z \pi(z_{ijkpqr});\notag\\
v_{ikjprq} \leftarrow v_{ikjprq} - \phi^z_1 \kappa^z \pi(z_{ijkpqr}); \notag\\
v_{jikqpr} \leftarrow v_{jikqpr} - \phi^z_2 \kappa^z \pi(z_{ijkpqr});\notag\\
v_{jkiqrp} \leftarrow v_{jkiqrp} - \phi^z_3 \kappa^z \pi(z_{ijkpqr});\notag\\
v_{kijrpq} \leftarrow v_{kijrpq} - \phi^z_4 \kappa^z \pi(z_{ijkpqr}); \notag\\
v_{kjirqp} \leftarrow v_{kjirqp} - \phi^z_5 \kappa^z \pi(z_{ijkpqr}).
\end{align}
Here, $0 \leq \kappa^z \leq 1$, $0\leq \phi^z_{\cdot} \leq 1$, and $\phi^z_1 + \phi^z_2 + \phi^z_3 + \phi^z_4 + \phi^z_5 = 1$. Similarly, for updating the dual multipliers $u_{ijpq}$ corresponding to constraints \eqref{eq:rlt1:con3}, we can write the following  rule.
\begin{align}
\label{eq:da:yrule1}
u_{ijpq} \leftarrow u_{ijpq} + \varphi \pi(y_{ijpq});\notag\\
u_{jiqp} \leftarrow u_{jiqp} - \varphi \pi(y_{ijpq}).
\end{align}
Here, $0\leq \varphi \leq 1$. We refer to it as ``Type 1 ascent rule.''

\item Now, let us consider the constraint \eqref{eq:dlrlt2:con2}. After applying Type 1 rule and solving the corresponding Z-LAP(s); for some $(i, j, p, q)$, it is possible that $\gamma_{ijp} + \delta_{ipq} - \Theta_{ijpq} < C_{ijpq}$, i.e., $\pi(y_{ijpq}) > 0$. In this case,  $\Theta_{ijpq}$ can be decreased by $\pi(y_{ijpq})$, by decreasing the cost coefficients $D_{ijkpqr}, \forall k, r$ by an amount $\frac{\pi(y_{ijpq})}{(n-2)}$. This allows us to increase the cost coefficients of the symmetrical variables, providing the objective functions of the corresponding LAPs a chance to grow. Mathematically, we adjust the  dual multipliers using the rule:
\begin{align}
\label{eq:da:rule2}
v_{ijkpqr} \leftarrow v_{ijkpqr} + \kappa^y \frac{\pi(y_{ijpq})}{(n-2)};\notag\\
v_{ikjprq} \leftarrow v_{ikjprq} - \phi^y_1 \kappa^y \frac{\pi(y_{ijpq})}{(n-2)};\notag\\
v_{jikqpr} \leftarrow v_{jikqpr} - \phi^y_2 \kappa^y \frac{\pi(y_{ijpq})}{(n-2)};\notag\\
v_{jkiqrp} \leftarrow v_{jkiqrp} - \phi^y_3 \kappa^y \frac{\pi(y_{ijpq})}{(n-2)};\notag\\
v_{kijrpq} \leftarrow v_{kijrpq} - \phi^y_4 \kappa^y \frac{\pi(y_{ijpq})}{(n-2)};\notag\\
v_{kjirqp} \leftarrow v_{kjirqp} - \phi^y_5 \kappa^y \frac{\pi(y_{ijpq})}{(n-2)}.
\end{align}
Here, $0 \leq \kappa^y \leq 1$, $0\leq \phi^y_{\cdot} \leq 1$, and $\phi^y_1 + \phi^y_2 + \phi^y_3 + \phi^y_4 + \phi^y_5 = 1$. We refer to it as ``Type 2 ascent rule.''

\item Finally, we can use a similar rule for constraint \eqref{eq:dlrlt2:con1}. For some $(i, p)$, if $\pi(x_{ip}) > 0$, we can decrease the cost coefficients $C_{ijpq}, \forall j, q$, by an amount of $\frac{\pi(x_{ip})}{(n-1)}$. This is equivalent to decreasing the cost coefficients $D_{ijkpqr}, \forall j, k, q, r$ by an amount $\frac{\pi(x_{ip})}{(n-1)(n-2)}$. Consequently, we can increase the cost coefficients of the symmetrical variables, potentially improving the objective function value of the corresponding LAPs. Mathematically, we adjust the  dual multipliers using the rule:
\begin{align}
\label{eq:da:rule3}
v_{ijkpqr} \leftarrow v_{ijkpqr} + \kappa^x \frac{\pi(x_{ip})}{(n-1)(n-2)};\notag\\
v_{ikjprq} \leftarrow v_{ikjprq} - \phi^x_1 \kappa^x \frac{\pi(x_{ip})}{(n-1)(n-2)};\notag\\
v_{jikqpr} \leftarrow v_{jikqpr} - \phi^x_2 \kappa^x \frac{\pi(x_{ip})}{(n-1)(n-2)};\notag\\
v_{jkiqrp} \leftarrow v_{jkiqrp} - \phi^x_3 \kappa^x \frac{\pi(x_{ip})}{(n-1)(n-2)};\notag\\
v_{kijrpq} \leftarrow v_{kijrpq} - \phi^x_4 \kappa^x \frac{\pi(x_{ip})}{(n-1)(n-2)};\notag\\
v_{kjirqp} \leftarrow v_{kjirqp} - \phi^x_5 \kappa^x \frac{\pi(x_{ip})}{(n-1)(n-2)}.
\end{align}
Here, $0 \leq \kappa^x \leq 1$, $0\leq \phi^x_{\cdot} \leq 1$, and $\phi^x_1 + \phi^x_2 + \phi^x_3 + \phi^x_4 + \phi^x_5 = 1$. We refer to it as ``Type 3 ascent rule.''

\item We can also implement a ``Type 4 ascent rule,'' in which we can generate two fractions $0 \leq \kappa^{lb}_{i'} \leq 1$ and $0 \leq \kappa^{lb}_{p'}\leq 1$ such that $(\kappa^{lb}_{i'} + \kappa^{lb}_{p'}) \leq 1$. Then we decrease the current lower bound $\nu(\text{LDRLT2})$ by the fraction $(\kappa^{lb}_{i'} + \kappa^{lb}_{p'})$, which is equivalent to decreasing the cost coefficients $b_{i'p}, \forall p$ by $\frac{\kappa^{lb}_{i'}\nu(\text{LDRLT2})}{n}$ and cost coefficients $b_{ip'}, \forall i$ by $\frac{\kappa^{lb}_{p'}\nu(\text{LDRLT2})}{n}$. This  is equivalent to decreasing the corresponding cost coefficients  $D_{i'jkpqr}, \forall j, k, p, q, r$ by an amount $\frac{\kappa^{lb}_{i'}\nu(\text{LDRLT2})}{n(n-1)(n-2)}$; and $D_{ijkp'qr}, \forall i, j, k, q, r$ by an amount $\frac{\kappa^{lb}_{p'}\nu(\text{LDRLT2})}{n(n-1)(n-2)}$. Consequently, we can increase the cost coefficients of the symmetrical variables, potentially improving the objective function values of the corresponding LAPs. This step deteriorates the current lower bound, however, the resulting redistribution provides a much greater increase in $\nu(\text{LDRLT2})$. This step can be implemented in the same spirit as the Simulated Annealing (SA) approach with a specific temperature schedule. \citet{hahn1998} reported stronger lower bounds for SA based dual ascent for RLT1, as compared to those of the naive dual ascent of \citet{adams1994}. Although it was not mentioned explicitly, we suspect that this approach was also used in dual ascent for RLT2 by \citet{adams2007}. In Section \ref{sec:ch4:comp}, we compare the lower bounds for various problems, with and without SA.

\item The overall step-size rule for Lagrangian dual ascent is a combination of the four rules discussed above. The solution complexity of the dual ascent phase is $O(n^6)$, which is same as the upper bound on the number of cost coefficients.

\end{enumerate}

\paragraph{Procedure RLT2-DA.} Once the dual multipliers are updated,  the LAPs need to be re-solved to obtain an improved $\nu(\text{LDRLT2})$, which is also a lower bound on the QAP. Thus the RLT2-DA procedure iterates between the LAP solution phase and the dual ascent phase, until a specified optimality gap has been achieved; a specified iteration limit has been reached; or a feasible solution to the QAP has been found. The steps of RLT2-DA procedure are depicted in Algorithm \ref{pro:rlt2da}.

\begin{algorithm}[htb]
\caption{RLT2-DA.}
\label{pro:rlt2da}
\begin{enumerate}
\item Initialization:
\begin{enumerate}
\item Initialize $m \leftarrow 0$, $\mathbf{v}^m \leftarrow \mathbf{0}$, $\bar{\nu}(\text{LDRLT2}) \leftarrow -\infty$, and $\text{GAP} \leftarrow \infty$.
\item Initialize $\mathbf{b}'$, $\mathbf{C}'$ and $\mathbf{D}'$.
\end{enumerate}
\item Termination: Stop if $m > \text{ITN\_LIM}$ or
$\text{GAP} < \text{MIN\_GAP}$ or
$\text{Feasibility check} = true$.
\item Z-LAP solve:
\begin{enumerate}
\item Update the dual multipliers $v^{m}_{ijkpqr} \leftarrow v^{m-1}_{ijkpqr} + \lambda_{ijkpqr}$. 
\item Update $D'_{ijkpqr} \leftarrow D'_{ijkpqr} - v^m_{ijkpqr}, \forall(i\neq j\neq k, p\neq q\neq r)$
\item Solve $n^2(n-2)^2$ Z-LAPs of size $(n-2)\times (n-2)$ and cost coefficients $\mathbf{D}'$.
\item Let $\Theta_{ijpq}(\mathbf{v}^m) \leftarrow \nu(\text{Z-LAP}(i, j, p, q)), \; \forall(i \neq j, p \neq q)$.
\item Update $C'_{ijpq} \leftarrow C'_{ijpq}  + \Theta_{ijpq}(\mathbf{v}^m), \forall (i \neq j, p \neq q)$.
\end{enumerate}
\item Y-LAP solve:
\begin{enumerate}
\item Update the dual multipliers $u^{m}_{ijpq} \leftarrow u^{m-1}_{ijpq} + \mu_{ijpq}$.
\item Update $C'_{ijpq} \leftarrow C'_{ijpq} - u^m_{ijpq}, \forall (i \neq j, p \neq q)$.
\item Solve $n^2$ Y-LAPs of size $(n-1)\times (n-1)$ and cost coefficients $\mathbf{C}'$.
\item Let $\Delta_{ip}(\mathbf{v}^m) \leftarrow \nu(\text{Y-LAP}(i, p)), \; \forall(i, p)$.
\end{enumerate}
\item X-LAP solve:
\begin{enumerate}
\item Update $b'_{ip} \leftarrow b_{ip} + \Delta_{ip}(\mathbf{v}^m), \; \forall (i, p)$.
\item Solve a single X-LAP of size $n\times n$ and cost coefficients $\mathbf{b}'$.
\item Update $\nu(\text{LRLT2}(\mathbf{v}^m)) \leftarrow \nu(\text{X-LAP})$.
\item If $\bar{\nu}(\text{LDRLT2}) < \nu(\text{LRLT2}(\mathbf{v}^m))$, update $\bar{\nu}(\text{LDRLT2}) \leftarrow \nu(\text{LRLT2}(\mathbf{v}^m))$ and $\text{GAP}$.
\end{enumerate}

\item Update $m\leftarrow m+1$. Return to Step 2. 

\end{enumerate}
\end{algorithm}

\paragraph{Feasibility Check.} To check whether the primal-dual feasibility has been achieved or not, the complementary slackness principle can be used. After solving the X-LAP and obtaining a primal solution $\mathbf{x}$, we construct feasible $\mathbf{y}$ and $\mathbf{z}$ vectors; and check whether the dual slack values $\bm{\pi}(\mathbf{y})$ and $\bm{\pi}(\mathbf{z})$ corresponding to this primal solution are compliant with Equation \eqref{eq:compslack}. If this is true, then a feasible solution has been found, which also happens to be optimal to the QAP.  Otherwise we continue to update the dual multipliers and re-solve LRLT2. 

\paragraph{Algorithm Correctness.} For the sake of completeness, we will now prove that the RLT2-DA provides us with a sequence of non-decreasing lower bounds on the QAP. This result has been adapted from the result by \citet{adams1994}.

\begin{theorem} \label{thm:dalb}
Given the input parameters $0 \leq \kappa \leq 1$, $0 \leq \phi \leq 1$, and $\sum \phi = 1$, the RLT2-DA provides a non-decreasing sequence of lower bounds.
\end{theorem}

\begin{proof}

Let us consider LRLT2 at some iterations $m$ and $m+1$, with the corresponding dual multipliers $\mathbf{v}^{m}$ and $\mathbf{v}^{m+1}$. To prove the theorem we need to show that $\nu(\text{LRLT2}(\mathbf{v}^{m+1})) \geq \nu(\text{LRLT2}(\mathbf{v}^{m}))$. Consider the following dual of $\text{LRLT2}(\mathbf{v}^{m+1})$. Note that we have not shown the conditions on the indices $i\neq j\neq k, p\neq q\neq r$ for the sake of brevity.
\begin{flalign}
\text{DLRLT2}(\mathbf{v}^{m+1}) = \max & \sum_{i} \alpha_i + \sum_{p} \beta_p; \\
\text{s.t. } & \alpha_i + \beta_p - \sum_{j\neq i} \gamma_{ijp} - \sum_{q \neq p} \delta_{ipq} \leq b_{ip}; \\
             & \gamma_{ijp} + \delta_{ipq} - \sum_{k \neq i,j} \xi_{ijkpq} - \sum_{r \neq p, q} \psi_{ijpqr} \leq C_{ijpq} - u^{m+1}_{ijpq};\\
              & \xi_{ijkpq} + \psi_{ijpqr} \leq D_{ijkpqr} - v^{m+1}_{ijkpqr};  \label{eq:ch4:proof:1} \\
             & \alpha_{i}, \beta_{p}, \gamma_{ijp}, \delta_{ipq}, \xi_{ijkpq}, \psi_{ijpqr} \sim \text{unrestricted}.&&
\end{flalign}
Let us assume, without the loss of generality, that $u^{m+1} = u^m$. We can substitute $v^{m+1}$ in Equation \eqref{eq:ch4:proof:1} with the following expression  arising from the three dual ascent rules \eqref{eq:da:rule1}, \eqref{eq:da:rule2}, and \eqref{eq:da:rule3}.
\begin{equation}
v^{m+1}_{ijkpqr} =  v^{m}_{ijkpqr} + \kappa^z \pi(z_{ijkpqr}) +  \frac{\kappa^y \pi(y_{ijpq})}{(n-2)} +   \frac{\kappa^x \pi(x_{ip})}{(n-1)(n-2)} - \Omega_{ijkpqr},  \label{eq:ch4:proof:2}
\end{equation}
where, $\Omega_{ijkpqr} \geq 0$ represents the sum of fractional slacks $\pi(x)$, $\pi(y)$, and $\pi(z)$, of the five symmetrical variables of $z_{ijkpqr}$, as given by the rules \eqref{eq:da:rule1}--\eqref{eq:da:rule3}. After substituting Equation \eqref{eq:ch4:proof:2} in Equation \eqref{eq:ch4:proof:1} and re-arranging the terms, we obtain the following constraint:
\begin{align}
\xi_{ijkpq} + \psi_{ijpqr} \leq & D_{ijkpqr} -  v^{m}_{ijkpqr} - \pi(z_{ijkpqr}) -  \frac{\pi(y_{ijpq})}{(n-2)} -  \frac{\pi(x_{ip})}{(n-1)(n-2)} \notag \\ 
	& + (1-\kappa^z)\pi(z_{ijkpqr}) +  \frac{(1-\kappa^y) \pi(y_{ijpq})}{(n-2)} +   \frac{(1-\kappa^x) \pi(x_{ip})}{(n-1)(n-2)}  + \Omega_{ijkpqr}.  \label{eq:ch4:proof:3}
\end{align}
After replacing Equation \eqref{eq:ch4:proof:1} with Equation  \eqref{eq:ch4:proof:3}, and aggregating the $\frac{\pi(y_{ijpq})}{(n-2)}$ and $\frac{\pi(x_{ip})}{(n-1)(n-2)}$ terms, we can write the following expression:
\begin{equation}
\nu(\text{LRLT2}(\mathbf{v}^{m+1})) = \nu(\text{DLRLT2}(\mathbf{v}^{m+1})) = \max_{\mathbf{\Psi}} \left\lbrace \sum_{i} \alpha_i + \sum_{p} \beta_p \right\rbrace,
\end{equation}
where, $\mathbf{\Psi}$ represents the constraint set:
\begin{align}
\alpha_i + \beta_p - \sum_{j\neq i} \gamma_{ijp} - \sum_{q \neq p} \delta_{ipq} \leq  \left [ b_{ip} - \pi(x_{ip}) \right ] + \left [ (1-\kappa^x) \pi(x_{ip}) \right ]; \\
\gamma_{ijp} + \delta_{ipq} - \sum_{k \neq i,j} \xi_{ijkpq} - \sum_{r \neq p, q} \psi_{ijpqr}  \leq \left [ C_{ijpq} - \pi(y_{ijpq}) - u^m_{ijpq} \right ] + \left [ (1-\kappa^y) \pi(y_{ijpq}) \right ];\\
\xi_{ijkpq} + \psi_{ijpqr} \leq \left [ D_{ijkpqr} - \pi(z_{ijkpqr}) - v^{m}_{ijkpqr} \right ]
              + \left [ (1-\kappa^z)\pi(z_{ijkpqr}) + \Omega_{ijkpqr} \right ];  \\
\alpha_{i}, \beta_{p}, \gamma_{ijp}, \delta_{ipq}, \xi_{ijkpq}, \psi_{ijpqr} \sim \text{unrestricted}.
\end{align}
If we split the constraint set $\mathbf{\Psi}$ into two constraint sets $\mathbf{\Psi_1}$ and $\mathbf{\Psi_2}$, such that,
\begin{align}
\mathbf{\Psi_1} \quad : \quad & \alpha_i + \beta_p - \sum_{j\neq i} \gamma_{ijp} - \sum_{q \neq p} \delta_{ipq} \leq  b_{ip} - \pi(x_{ip}); \\
& \gamma_{ijp} + \delta_{ipq} - \sum_{k \neq i,j} \xi_{ijkpq} - \sum_{r \neq p, q} \psi_{ijpqr}  \leq C_{ijpq} - \pi(y_{ijpq}) - u^m_{ijpq};\\
& \xi_{ijkpq} + \psi_{ijpqr} \leq D_{ijkpqr}  - \pi(z_{ijkpqr}) - v^{m}_{ijkpqr};  \\
& \alpha_{i}, \beta_{p}, \gamma_{ijp}, \delta_{ipq}, \xi_{ijkpq}, \psi_{ijpqr} \sim \text{unrestricted};
\end{align}
\begin{align}
\mathbf{\Psi_2} \quad : \quad & \alpha_i + \beta_p - \sum_{j\neq i} \gamma_{ijp} - \sum_{q \neq p} \delta_{ipq} \leq  (1-\kappa^x) \pi(x_{ip}); \\
& \gamma_{ijp} + \delta_{ipq} - \sum_{k \neq i,j} \xi_{ijkpq} - \sum_{r \neq p, q} \psi_{ijpqr}  \leq (1-\kappa^y) \pi(y_{ijpq});\\
& \xi_{ijkpq} + \psi_{ijpqr} \leq (1-\kappa^z)\pi(z_{ijkpqr}) + \Omega_{ijkpqr};  \\
& \alpha_{i}, \beta_{p}, \gamma_{ijp}, \delta_{ipq}, \xi_{ijkpq}, \psi_{ijpqr} \sim \text{unrestricted};
\end{align}
then, from the theory of linear programming, we can show that:
\begin{equation}
\max_{\mathbf{\Psi}} \left \lbrace \sum_{i} \alpha_i + \sum_{p} \beta_p \right \rbrace \geq  \max_{\mathbf{\Psi_1}} \left  \lbrace \sum_{i} \alpha_i + \sum_{p} \beta_p \right \rbrace + \max_{\mathbf{\Psi_2}} \left \lbrace \sum_{i} \alpha_i + \sum_{p} \beta_p \right \rbrace .\label{eq:ch4:proof:4}
\end{equation}
Finally, we can assert that:  $\nu(\text{LRLT2}(\mathbf{v}^{m})) = \nu(\text{DLRLT2}(\mathbf{v}^{m})) = \max_{\Psi_1}  \left\lbrace \sum_{i} \alpha_i + \sum_{p} \beta_p \right\rbrace$, and due to the non-negativity of $\pi(x), \pi(y), \pi(z)$, we have: $ \max_{\Psi_2} \left\lbrace \sum_{i} \alpha_i + \sum_{p} \beta_p \right \rbrace \geq 0$. Therefore, 
\begin{equation}
\nu(\text{LRLT2}(\mathbf{v}^{m+1}))  \geq \nu(\text{LRLT2}(\mathbf{v}^{m})).
\end{equation}
Hence proved.
\end{proof}


\section{Accelerating RLT2-DA Algorithm Using a GPU Cluster} \label{sec:prlt2da}

The RLT2-DA algorithm described in the previous section was shown to outperform the Lagragian subgradient search and  many other algorithms in a branch-and-bound scheme, in terms of lower bound strength and the number of nodes fathomed, for problems with $n \leq 30$. However, for solving a QAP of size $n$ using RLT2-DA, we need to solve $O(n^4)$ LAPs of size $O(n^2)$, and update $O(n^6)$ dual multipliers. The overall solution complexity of sequential RLT2-DA is $O(n^7)$. An important observation about RLT2-DA is that the $O(n^4)$ LAPs can be solved independently of each other and similarly, the $O(n^6)$ Lagrangian multipliers can be updated independently of each other. Therefore, with the help of a correct parallel programming architecture, it is possible to achieve significant speedup over the sequential algorithm.

In the parallel algorithm that we have implemented, both phases of the sequential algorithm are executed on the GPU(s) by one or more CUDA kernels. 
We have chosen CUDA enabled NVIDIA GPUs as our primary architecture, because a GPU offers a large number of processor cores which can process a large number of threads in parallel. This is extremely useful for the dual update phase of RLT2-DA, since one CUDA thread can be assigned to each multiplier, and a host of multipliers can be updated simultaneously. Additionally, our efficient GPU-accelerated algorithm for the LAP can be used to speed up the LAP solution phase of RLT2-DA. Both these algorithms can be combined into a GPU-accelerated RLT2-DA solver engine, which can obtain strong lower bounds on the QAP, in an efficient manner. 


In the single GPU implementation of RLT2-DA solver, it may become challenging to store the $O(n^6)$ cost coefficients in the GPU memory, especially for larger problems. One of the alternatives to overcome this problem is to store the matrices in the CPU memory and copy them into the GPU memory as required. However, this approach requires $O(n^6)$ transfer between the CPU and GPU, which may introduce severe communication overhead. A better alternative is to split the cost coefficients (or LAP matrices) across multiple GPUs (if available), which also allows us to solve several LAPs concurrently. In this work, we have used grid architecture with multiple Processing Elements (PE), each containing one CPU-GPU pair. Communication between the different CPUs is accomplished using {\em message passing interface} (MPI). The overall algorithmic architecture is shown in Fig.\ \ref{fig:rlt2par} and the details of our implementation are described in the following sections.

\begin{figure}
\begin{center}
\includegraphics[width=0.85\columnwidth, page=3]{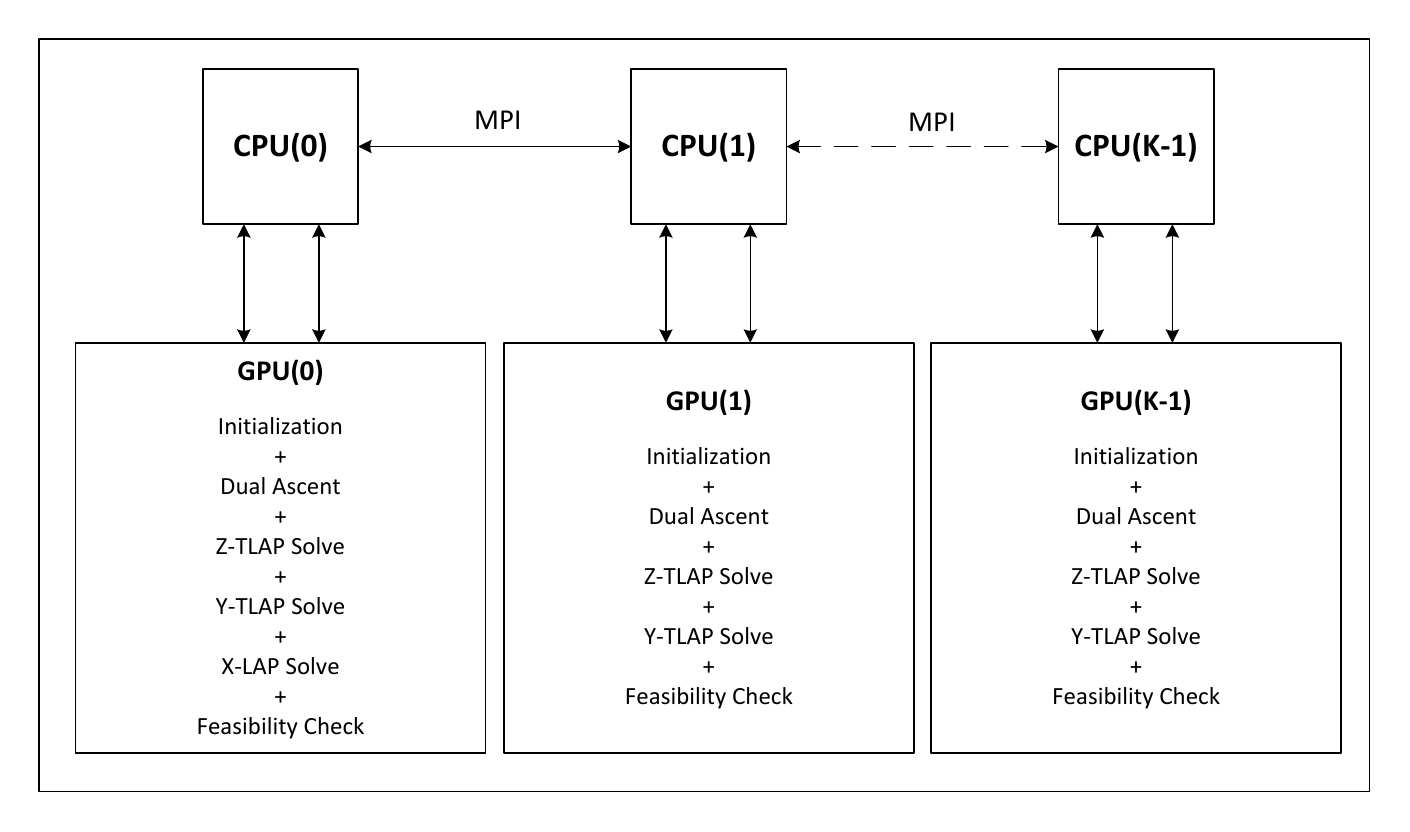}
\caption{Parallel/accelerated RLT2-DA.}
\label{fig:rlt2par}
\end{center}
\end{figure}

\subsection{Initialization} 

The program is initialized with $K$ MPI processes, equal to the number of PEs in the grid. It is assumed that one MPI process gets allocated to exactly one CPU. The CPU with rank 0 is chosen as the root. The cost matrices for the Y and Z-LAPs are split evenly across all the GPUs in the grid, i.e.,\ each device owns $M_z = \left \lceil \frac{n^2(n-1)^2}{K} \right \rceil$ number of Z-LAP matrices and $M_y = \left \lceil \frac{n^2}{K} \right \rceil$ number of Y-LAP matrices. The X-LAP matrix is owned only by the root GPU. 

\subsection{LAP Solution} 

All the LAPs are solved on the GPU using the {\em alternating-tree variant} of our GPU-accelerated Hungarian algorithm \citep{date2016gpu}. We know that our GPU-accelerated Hungarian algorithm is extremely efficient in solving large LAPs, rather than small LAPs. 
Therefore, all the LAPs owned by a particular GPU are combined and solved as a {\em tiled} LAP (or TLAP).  For example, if we have $M_z$ number of Z-LAP matrices (of size $(n-2)\times(n-2))$ on a particular GPU, then instead of solving them one at a time, we stack these LAP matrices and solve a single TLAP of size $M_z\times(n-2)\times (n-2)$.  The solution complexity for a TLAP is $O(M^{1.5}n^3)$, which is asymptotically worse than $O(Mn^3)$. However, in practice, we found that a single TLAP solves much faster than solving individual LAPs one by one. We suspect that this happens because of the execution overhead incurred due to repeated invocation of the CUDA kernels in the one-at-a-time approach versus invoking the kernels only once in the tiled approach.   

\subsection{Dual Ascent}

As mentioned earlier, Lagrangian multiplier update is quite straightforward to parallelize on a GPU. There is exactly one dual multiplier associated with one of the $z$ variables, and we can easily assign one GPU thread per element of the LAP cost matrices residing on a particular GPU. It is important to note that we do not need any additional data structures to store the dual multipliers $\mathbf{v}$, since we only need the updated cost coefficients $\mathbf{b}'$, $\mathbf{C}'$, and $\mathbf{D}'$ during  each iteration of RLT2-DA. Therefore, during the multiplier update step, these cost coefficients can be updated in-place with the specified ascent rule(s). 

For updating the dual multipliers (or cost coefficients), we need the dual slacks $\bm{\pi}(\mathbf{x})$, $\bm{\pi}(\mathbf{y})$, and $\bm{\pi}(\mathbf{z})$ of the symmetrical variables during each iteration, which might not be native to a particular GPU. In short, before each iteration, we need to transfer the arrays of dual variables $\bm{\alpha}, \bm{\beta}, \bm{\gamma}, \bm{\delta}, \bm{\xi}, \bm{\psi}$ and the modified cost matrices $\mathbf{b}'$, $\mathbf{C}'$, and $\mathbf{D}'$ between the various CPUs/GPUs, using MPI. This might incur a significant MPI communication overhead, since the matrix $\mathbf{D}'$ contains $O(n^6)$ elements. To alleviate this overhead, 
local copies of the symmetrical cost coefficients are stored on a GPU, for each of the $M_z(n-2)^2$ number of $D_{ijkpqr}$ cost coefficients owned by  that GPU. These local coefficients are kept up-to-date using the specified ascent rule. Therefore, we only need to transfer the dual variables $\bm{\alpha}, \bm{\beta}, \bm{\gamma}, \bm{\delta}, \bm{\xi}, \bm{\psi}$ and the cost matrices $\mathbf{b}', \mathbf{C}'$. 
This approach involves some duplication of work, however the communication complexity is reduced from $O(n^6)$ to $O(n^5)$. 

\subsection{Accelerated RLT2-DA and Variants}

The parallel algorithm for RLT2-DA is depicted in Algorithm \ref{pro:accrlt2da}. Most of the steps in this algorithm are the same as that of the sequential algorithm, with the exception that the LAP solution and the dual update phases are performed on the GPU, and MPI communication steps are added. 

\begin{algorithm}
\caption{Accelerated RLT2-DA.}
\label{pro:accrlt2da}
\begin{enumerate}
\item Initialization:
\begin{enumerate}
\item Initialize $m \leftarrow 0$, $\mathbf{v}^m \leftarrow \mathbf{0}$, $\bar{\nu}(\text{LDRLT2}) \leftarrow -\infty$, and $\text{GAP} \leftarrow \infty$.
\item Initialize $\mathbf{b}'$ on GPU(0).  Initialize $\mathbf{C}'$ and $\mathbf{D}'$ on respective GPUs.
\end{enumerate}
\item Termination: Stop if $m > \text{ITN\_LIM}$ or
$\text{GAP} < \text{MIN\_GAP}$ or
$\text{Feasibility check} = true$.

\item Z-LAP solve (parallely on $K$ GPUs):
\begin{enumerate}
\item Update the dual multipliers $v^{m}_{ijkpqr} \leftarrow v^{m-1}_{ijkpqr} + \lambda_{ijkpqr}$
\item Update $D'_{ijkpqr} \leftarrow D'_{ijkpqr} - v^m_{ijkpqr}, \forall(i\neq j\neq k, p\neq q\neq r)$
\item Solve Z-TLAP of size $M_z \times (n-2)\times (n-2)$ and cost coefficients $\mathbf{D}'$.
\item Let $\Theta_{ijpq}(\mathbf{v}^m) \leftarrow \nu(\text{Z-LAP}(i, j, p, q)), \; \forall(i \neq j, p \neq q)$.
\item Broadcast $\Theta_{ijpq}(\mathbf{v}^m)$, $\bm{\xi}_{ijpq}$, and $\bm{\psi}_{ijpq}$ to all CPUs/GPUs using \texttt{MPI\_Bcast}.
\item Update $C'_{ijpq} \leftarrow C'_{ijpq} + \Theta_{ijpq}(\mathbf{v}^m), \forall (i \neq j, p \neq q)$.
\end{enumerate}
\item Y-LAP solve (parallely on $K$ GPUs):
\begin{enumerate}
\item Update the dual multipliers $u^{m}_{ijpq} \leftarrow u^{m-1}_{ijpq} + \mu_{ijpq}$.
\item Update $C'_{ijpq} \leftarrow C'_{ijpq} - u^{m}_{ijpq}, \forall (i \neq j, p \neq q)$.
\item Solve Y-TLAP of size $M_y \times (n-1)\times (n-1)$ and cost coefficients $\mathbf{C}'$.
\item Let $\Delta_{ip}(\mathbf{v}^m) \leftarrow \nu(\text{Y-LAP}(i, p)), \; \forall(i, p)$.
\item Broadcast $\mathbf{C}'$, $\Delta_{ip}(\mathbf{v}^m)$, $\bm{\gamma}_{ip}$, and $\bm{\delta}_{ip}$ to all CPUs/GPUs using \texttt{MPI\_Bcast}.
\end{enumerate}
\item X-LAP solve (only on GPU(0)):
\begin{enumerate}
\item Update $b'_{ip} \leftarrow b_{ip} + \Delta_{ip}(\mathbf{v}^m), \; \forall (i, p)$.
\item Solve a single X-LAP of size $n\times n$ and cost coefficients $\mathbf{b}'$.
\item Update $\nu(\text{LRLT2}(\mathbf{v}^m)) \leftarrow \nu(\text{X-LAP})$.
\item If $\bar{\nu}(\text{LDRLT2}) < \nu(\text{LRLT2}(\mathbf{v}^m))$, update $\bar{\nu}(\text{LDRLT2}) \leftarrow \nu(\text{LRLT2}(\mathbf{v}^m))$ and $\text{GAP}$.
\item Broadcast $\bar{\nu}(\text{LDRLT2})$, GAP, $\mathbf{b}'$, $\bm{\alpha}$, and $\bm{\beta}$ to all CPUs/GPUs using \texttt{MPI\_Bcast}.
\end{enumerate}

\item Update $m\leftarrow m+1$. Return to Step 2. 

\end{enumerate}
\end{algorithm}

Two variants of the accelerated RLT2-DA algorithm are implemented, namely {\bf slow} RLT2-DA (S-RLT2-DA) and {\bf fast} RLT2-DA (F-RLT2-DA), which are based on Equation \eqref{eq:ch4:proof:4}. 
In the S-RLT2-DA variant, the LAPs with updated cost coefficients $\mathbf{b}'$, $\mathbf{C}'$, and $\mathbf{D}'$ are solved during each iteration, which corresponds to the left hand side of Equation \eqref{eq:ch4:proof:4}. The steps mentioned in Algorithm \ref{pro:accrlt2da} are essentially those of S-RLT2-DA. In the F-RLT2-DA variant, the LAPs are solved with the incremental cost coefficients (the fractional dual slacks $\bm{\pi}(\mathbf{x}), \bm{\pi}(\mathbf{y}), \bm{\pi}(\mathbf{z})$) from the right hand side of Equation \eqref{eq:ch4:proof:4}, and the result is added to the lower bound obtained during the previous iteration. This variant has smaller execution time (hence the name {\em fast}), since the incremental cost coefficient matrices are sparser as compared to the actual cost coefficient matrices. 
However, as the inequality suggests, the lower bound of S-RLT2-DA is stronger than that of F-RLT2-DA. Another advantage of using S-RLT2-DA is that during each iteration, we have the updated dual multipliers (in the form of cost coefficients $\mathbf{b}', \mathbf{C}', \mathbf{D}'$) which can be used as a starting solution for the children nodes in a branch-and-bound scheme. However, in F-RLT2-DA, recovering the actual dual multipliers is not so straightforward. 

For both the above variants, a stronger lower bound can be obtained by adopting a two-phase approach. In the first phase (Step 3b-1), Z-TLAPs with the cost coefficients $\mathbf{D}'$ are solved. During the second phase, initially (Step 3b-2), for each $(i < j < k, p\neq q\neq r)$, the six symmetrical $z$ variables are partitioned into two sets based on their dual slacks: $S_B = \{z: \pi(z) = 0\}$ and $S_N = \{z: \pi(z) > 0\}$. Then, the dual slacks of the symmetrical variables from $S_N$ are added; their cost coefficients are reduced by the corresponding $\pi(z)$; and the sum is evenly distributed across the cost coefficients of the symmetrical variables from $S_B$. Finally (Step 3b-3), the Z-TLAPs are re-solved and the algorithm continues to Step 3c. Since we are solving the TLAPS two times, this two-phase approach takes almost twice the time of the one-phase approach, but it provides the strongest lower bounds. 

A third variant is to use Simulated Annealing with some temperature schedule, in which the algorithm is allowed to redistribute a random fraction of the current lower bound among some of the $z$ variables (see Type 4 ascent rule in Section \ref{sec:rlt2da}). This deteriorating step provides the algorithm with an opportunity to get out of a local maximum, which further improves the lower bound.

In Section \ref{sec:ch4:comp}, we will compare the lower bounds and execution times for each of the above variants, which will provide significant insight to practitioners to make careful selection. We will refer to these variants as F1, F2, S1, S2, where the first letter is used for distinction between fast and slow variants, while the number indicates whether it is a single-phase or two-phase algorithm. 




\section{Parallel Branch-and-bound with Accelerated RLT2-DA} \label{sec:ch4:bnb}

Although the objective function value of the LP relaxation of RLT2 was shown to  provide tight lower bound (equal to the integer optimal) for  small QAPs ($n \leq 12$), the LP relaxation is expected to have a duality gap for medium and large QAPs. Also, due to the integrality of LRLT2$(\mathbf{v})$,  $\nu(\text{LDRLT2})$ can only ever reach the LP relaxation bound. Therefore, the LDRLT2 objective function value provides a lower bound on both LPRLT2 and QAP objective function values, and RLT2-DA cannot be used on its own to find exact solutions to large QAPs. To this end,  we need to use a branch-and-bound (B\&B) procedure to solve medium and large QAPs to optimality.

B\&B is a standard tree search procedure used for solving integer optimization problems. Each node in the B\&B tree corresponds to a subproblem in which a single variable is assigned a specific value. This assignment partitions the solution space into two or more disjoint subspaces.  Solving an LP relaxation of the subproblem at a particular node provides a lower bound on that node. A node and its children are fathomed if any one of the following three conditions are satisfied: (1) The lower bound at that node exceeds or equals the incumbent objective value; (2) The subproblem is infeasible; or (3) The subproblem has an integral solution. Fathomed nodes are not considered for further exploration and the whole branch is removed from the tree. Therefore, quality of the lower bound is of utmost importance, so as to explore as few nodes as possible. Another important consideration in B\&B scheme is the search strategy to be employed for exploring the tree. The tree could be explored using Breadth-First-Search (BFS), or Depth-First-Search (DFS), or Best-First-Search (BstFS), each of which has its own pros and cons. We used a hybrid BstFS+DFS approach, since it was more suitable for the problem under study.

The specifics of our parallel B\&B procedure  can be explained as follows. At the root level of the B\&B tree, none of the facility locations are fixed. At each subsequent level $\ell$ of the B\&B tree, the locations of some $(\ell-1)$ facilities are fixed. The $\ell^{th}$ facility is assigned to each one of the remaining $(n-\ell+1)$ locations, giving rise to $(n-\ell+1)$ children nodes at that level. This type of branching is called ``polytomic'' branching and it has been used for solving QAPs  by \citet{roucairol1987}, \citet{pardalos1989}, \citet{clausen1997}, and \citet{anstreicher2002} using other formulations and other lower bounding techniques.

For each node of the B\&B tree, the lower bound is obtained using our accelerated RLT2-DA executed by a bank of PEs (CPU-GPU pairs). Figure \ref{fig:rlt2par} depicts one such PE bank. Performing RLT2-DA on the root node produces the root lower bound (reported in Section \ref{sec:ch4:comp}). Parallelizing the B\&B procedure is quite straightforward since we can allocate multiple banks of PEs, such that each bank is responsible for a subset of unexplored nodes and corresponding sub-trees. Load balancing is a critical aspect of parallel B\&B so as to improve processor utilization. Load balancing can be achieved by precisely managing the queue of unexplored nodes and redistributing them over the idle PE banks as required. Figure \ref{fig:ch4:parbnb} shows the architecture used in our  parallel B\&B scheme.

\begin{figure}
\begin{center}
\includegraphics[width=\columnwidth, page=7]{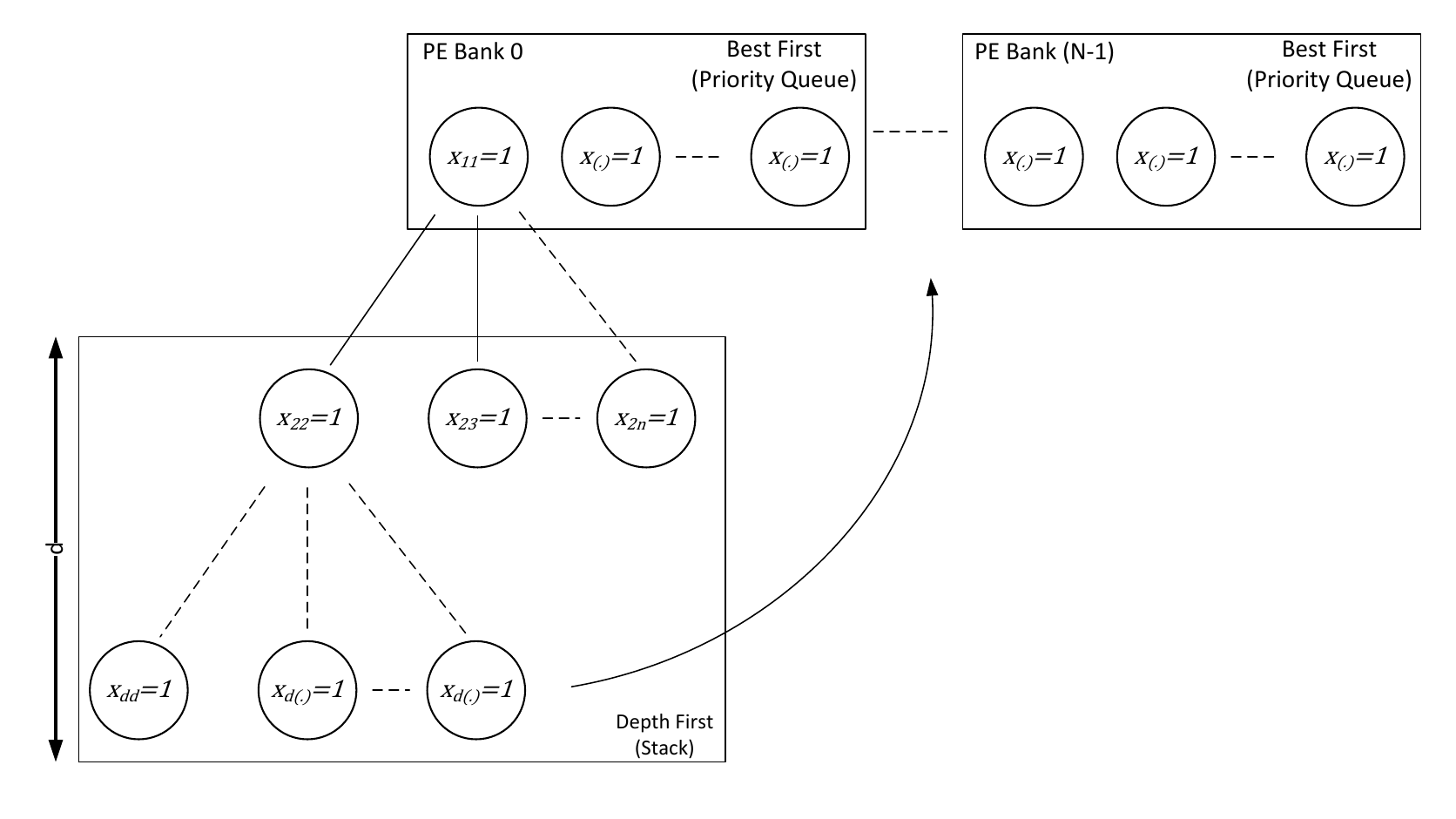}
\caption{Parallel branch-and-bound.}
\label{fig:ch4:parbnb}
\end{center}
\end{figure}

\begin{enumerate}

\item We begin the parallel B\&B with $N$ banks containing $K$ PEs each. CPU0 is designated as the ``Master Processor'' (MP),  which maintains a ``master list'' of unexplored nodes such that the node with the smallest lower bound is at the top of the list (essentially a ``heap''). This makes sure that the first PE Bank is always working on the nodes with the best bound (essentially a BstFS). The master list is seeded with some initial set of nodes from some level $\ell_{\text{init}}$. As an example, if $\ell_{\text{init}} = 0$, then the list contains only the root node in which none of the facilities are assigned to any of the locations. For $\ell_{\text{init}} = 1$, there will be $n$ nodes in which one of the facilities is assigned to all the $n$ locations.

\item The nodes from the master list are equally distributed among all the PEs, which explore the respective sub-trees in a DFS manner. In some of the QAP instances (e.g., Nugent instances) the locations are present on a grid, which allows us to apply symmetry elimination rules and eliminate a number of assignments which will have the same objective value. 

\item For each of the nodes allocated to a particular PE bank, a lower bound is calculated using our accelerated RLT2-DA. 
If a node at level $\ell$ cannot be fathomed 
, then it is branched upon by placing another $(\ell + 1)^{th}$ facility at all the available locations and generating $(n - \ell)$ children. For all these children, the dual multipliers of the parent node can be used as an initial solution (warm start), which saves us from calculating the dual multipliers from scratch. This significantly speeds up RLT2-DA execution.

\item The order in which the facilities (rows) are considered for branching is extremely important to ensure limited exploration of the B\&B tree. We experimented with the following  simple rules to generate this placement order ahead of time. We found that Rule 3 performs the best in terms of number of nodes explored. 
\begin{description}
\item[Rule 1:] The facilities (rows) to be placed are considered according to their indices, i.e., facility (row) 1 is placed first, facility (row) 2 is placed next, and so on. 
\item[Rule 2:] The first facility (row) to be placed is the one which has the highest total interaction; the second facility (row) to be placed is the one which has the highest total interaction with the  first facility (row); and so on. 
\item[Rule 3:] The first facility (row) to be placed is the one which has the lowest total interaction; the second facility (row) to be placed is the one which has the lowest total interaction with the  first facility (row); and so on. 
\end{description}

\item An important aspect of B\&B is to distinguish between the amount of time spent in improving the lower bound versus branching, which will produce a number of children with improved lower bounds. For this purpose, some early termination criteria, based on the observed improvement in the lower bound, can be used. If the observed improvement is not sufficient, then  RLT2-DA iterations are stopped and the node is branched upon. 

\item While the DFS strategy keeps all the PE banks fairly busy, it is possible that some PE banks may obtain ``easy'' nodes which can be fathomed fairly quickly. In such cases, those PE banks will remain idle, which is detrimental for the system utilization. Therefore, a load balancing scheme can be implemented similar to the one implemented by \citet{anstreicher2002}. In this scheme, the idle PE bank sends a request to the MP. The MP checks the request queue after every 300 seconds. If there is at least one processor that is idle, then  the unexplored nodes from all the PE banks are collected by the MP and redistributed evenly across all the PEs.

\item Finally, it may be beneficial to limit the DFS exploration up to a maximum depth $d$. This is because, we need to save the $O(n^6)$ dual multipliers associated with each node at a particular level, to be able to perform warm start on its children nodes. These multipliers are saved in the CPU memory of the corresponding PEs from the bank, and for depth $d$, the space complexity becomes $O(d\cdot n^6)$. 
Any unexplored nodes beyond the depth $d$ are collected and redistributed by the MP, and the memory is reset.

\end{enumerate}

The computational results for this parallel B\&B scheme coupled with the accelerated RLT2-DA procedure are presented in the next section.


\section{Computational Experiments}\label{sec:ch4:comp}

The accelerated RLT2-DA algorithm was coded in C++ and CUDA C programming languages and deployed on the Blue Waters Supercomputing Facility at the University of Illinois at Urbana-Champaign. Each PE consists of an AMD Interlagos model 6276 CPU, with 8 cores, 2.3GHz clock speed, and 32GB memory; connected to an NVIDIA GK110 ``Kepler" K20X GPU, with 2688 processor cores, and 6GB memory. 

Various computational studies were conducted on different variants of our parallel/accelerated RLT2-DA, with respect to the bound strength, scaling behavior, and performance in branch-and-bound procedure. For testing purposes, we used various solved and unsolved instances of size $20 \leq n \leq 42$ from the QAPLIB \citep{qaplib1997}. These computational tests are presented in the following sections.

As documented by \citet{goncalves2017}, with smart selection of the dual ascent parameters $\kappa$, $\phi$, and $\varphi$, we can get away with solving only half the Z-LAPs, which results in significant savings in both GPU memory and time. Specifically, for some $i< j < k, p\neq q\neq r$, we use: $\kappa^y = \kappa^x = 1$ for all six variables; $\kappa^z = 1$ and $\phi^z = \phi^y = \phi^x = 0$ for the lower order variables $z_{jikqpr}, z_{kijrpq}, z_{kjirqp}$; and $\kappa^z = \frac{2}{3}$ and $\phi^z = \phi^y = \phi^x = \frac{1}{2}$ for the upper order  variables $z_{ijkpqr}, z_{ikjprq}, z_{jkiqrp}$. This means that the partial dual slacks are  split only among the three upper order variables and the lower order variables will have 0 slack at the end of the ascent step. As a result, the Z-LAPs containing these variables, specifically Z-LAP($i, j, p ,q$), $\forall (i > j, p\neq q)$, can be omitted. After solving Z-LAPs and updating $C_{ijpq}, \forall i < j, p \neq q$, the multipliers $u_{ijpq}$ need to be adjusted using $\varphi = \frac{1}{2}$, which means that the updated dual slack $\pi(y_{ijpq})$ is split equally among $y_{ijpq}$ and $y_{jiqp}$.
 
For the SA based variants, the following annealing schedule was used. The initial temperature $T$ was set to $4\%$ of the best known upper bound. The fractions $\kappa^{lb}_i$ and $\kappa^{lb}_p$ were generated randomly $\forall i, p$; with a constraint that $\kappa^{lb} = \sum_i \kappa^{lb}_i + \sum_p \kappa^{lb}_p \leq 0.25$. This means that at most 25\% of the current lower bound was made available for redistribution using Type 4 ascent rule. The acceptance probability is calculated using the formula: $p_{acc} = \exp(-\kappa^{lb}/T)$. The value $\kappa^{lb}$ was redistributed if a randomly generated number $p_{rand} \leq p_{acc}$. The temperature is reduced by a factor of 0.99 after every 100 iterations, to make sure that the acceptance probability of redistributing a larger $\kappa^{lb}$ decreases with increasing number of iterations. 

Since accelerated RLT2-DA algorithm is memory intensive, instances of specific size requires a certain minimum number of GPUs to be able to fit all the necessary data structures. Table \ref{tbl:ch4:mingpu} lists the minimum number of PEs required to be able to comfortably solve the QAP instances of different sizes. Note that these numbers are derived according to the specifications of GPUs that we used for testing. These numbers might change if we use GPUs with specifications other than the ones mentioned earlier.

\begin{table}
\begin{center}
\caption{Minimum number of PEs.}
\label{tbl:ch4:mingpu}
{\begin{tabular}{|c|c|c|c|c|c|}
\hline
n & $\leq 27$ & 30 & 35 & 40 & 42\tabularnewline
\hline
Minimum \# of PEs & 1 & 2 & 4 & 10 & 15 \tabularnewline
\hline
\end{tabular}
}
\end{center}
\end{table}


\subsection{Comparison of RLT2-DA Variants}

To compare the strength of lower bounds, we used the Nug20 instance (which has $n=20$ facilities and locations) from the QAPLIB \citep{qaplib1997}. On this instance, we ran 2000 iterations of the different variants. The tests were performed with only a single PE. All the Z-LAPs were tiled into a single Z-TLAP, and all the Y-LAPs were tiled into a single Y-TLAP. For these tests, we noted the lower bounds and execution times. For the first iteration, i.e., for $\bm{v} = \bm{0}$, we obtain the Gilmore-Lawler bound of 2057. After that, RLT2-DA obtains an increasing sequence of lower bounds during the subsequent iterations, in accordance with Theorem \ref{thm:dalb}. The results are shown in Table \ref{tbl:ch4:nug20}, Fig.\  \ref{fig:ch4:nug20lb}, and Fig.\  \ref{fig:ch4:nug20times}. From these results, we can draw the following conclusions:

\begin{table}
\begin{center}
\caption{Bound strength of RLT2-DA variants on Nug20.}
\label{tbl:ch4:nug20}
{\resizebox{\columnwidth}{!}{%
\begin{tabular}{ccccccccccccc}
\hline
\multirow{2}{*}{Itn} & \multicolumn{3}{c}{F1} & \multicolumn{3}{c}{F2} & \multicolumn{3}{c}{S1} & \multicolumn{3}{c}{S2}\tabularnewline
\cline{2-13}
 & w/o SA & w/ SA & Time (s) & w/o SA & w/ SA & Time (s) & w/o SA & w/ SA & Time (s) & w/o SA & w/ SA & Time (s)\tabularnewline
\cline{1-13}
200 & 2440.87 & 2483.62 & 297.55 & 2454.62 & 2492.43 & 1322.39 & 2445.59 & 2479.24 & 1081.25 & 2457.28 & 2488.07 & 2141.85\tabularnewline

400 & 2444.51 & 2493.76 & 594.68 & 2458.94 & 2503.20 & 2621.34 & 2448.92 & 2488.42 & 2256.91 & 2460.47 & 2497.51 & 4386.91\tabularnewline

600 & 2445.66 & 2498.45 & 891.91 & 2460.54 & 2507.43 & 3915.88 & 2450.08 & 2492.49 & 3479.37 & 2461.57 & 2501.13 & 6687.30\tabularnewline

800 & 2446.22 & 2501.62 & 1192.79 & 2461.41 & 2509.69 & 5204.72 & 2450.72 & 2494.80 & 4693.84 & 2462.18 & 2502.41 & 8983.03\tabularnewline

1000 & 2446.55 & 2502.80 & 1492.82 & 2461.96 & 2511.04 & 6489.81 & 2451.14 & 2496.13 & 5945.78 & 2462.57 & 2504.03 & 11325.10\tabularnewline

1200 & 2446.78 & 2504.97 & 1789.72 & 2462.35 & 2512.74 & 7776.94 & 2451.44 & 2497.85 & 7164.30 & 2462.85 & 2505.36 & 13645.90\tabularnewline

1400 & 2446.96 & 2506.44 & 2083.91 & 2462.65 & 2513.97 & 9061.53 & 2451.67 & 2499.04 & 8417.98 & 2463.06 & 2506.34 & 15998.50\tabularnewline

1600 & 2447.12 & 2508.24 & 2385.15 & 2462.89 & 2515.40 & 10338.90 & 2451.85 & 2500.41 & 9703.68 & 2463.24 & 2507.37 & 18375.60\tabularnewline

1800 & 2447.26 & 2509.15 & 2680.13 & 2463.08 & 2516.28 & 11598.70 & 2452.00 & 2501.39 & 11042.30 & 2463.39 & 2508.25 & 20805.10\tabularnewline

2000 & 2447.39 & 2510.04 & 2980.25 & 2463.24 & 2517.07 & 12861.40 & 2452.12 & 2502.10 & 12368.30 & 2463.51 & 2508.90 & 23225.70\tabularnewline
\hline
\end{tabular}
}
}
\end{center}
\end{table}

\begin{figure}
\begin{center}
\includegraphics[width=0.75\columnwidth,page=13]{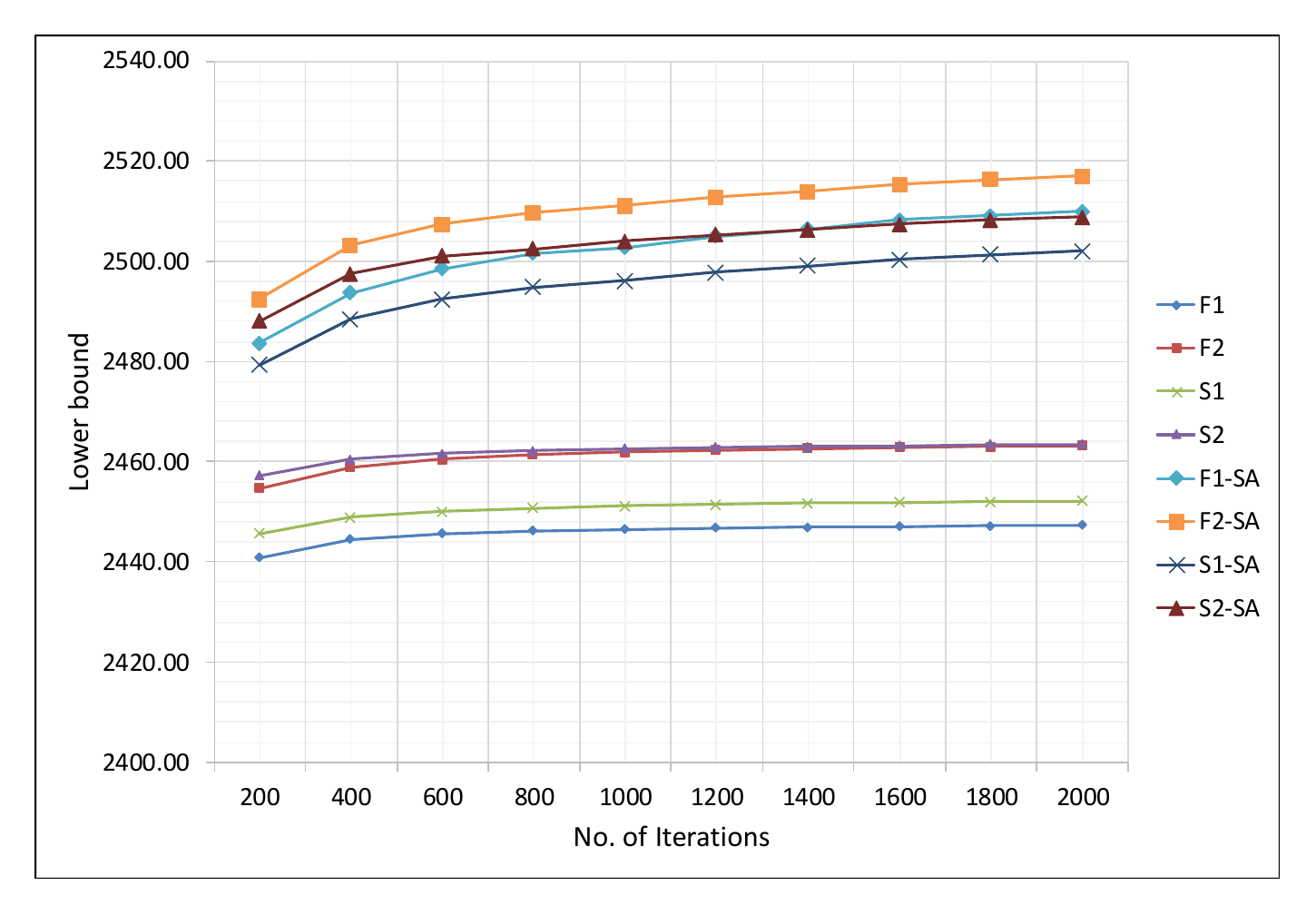}
\caption{Lower bounds of accelerated RLT2-DA variants.}
\label{fig:ch4:nug20lb}
\end{center}
\end{figure}

\begin{figure}
\begin{center}
\includegraphics[width=0.75\columnwidth,page=14]{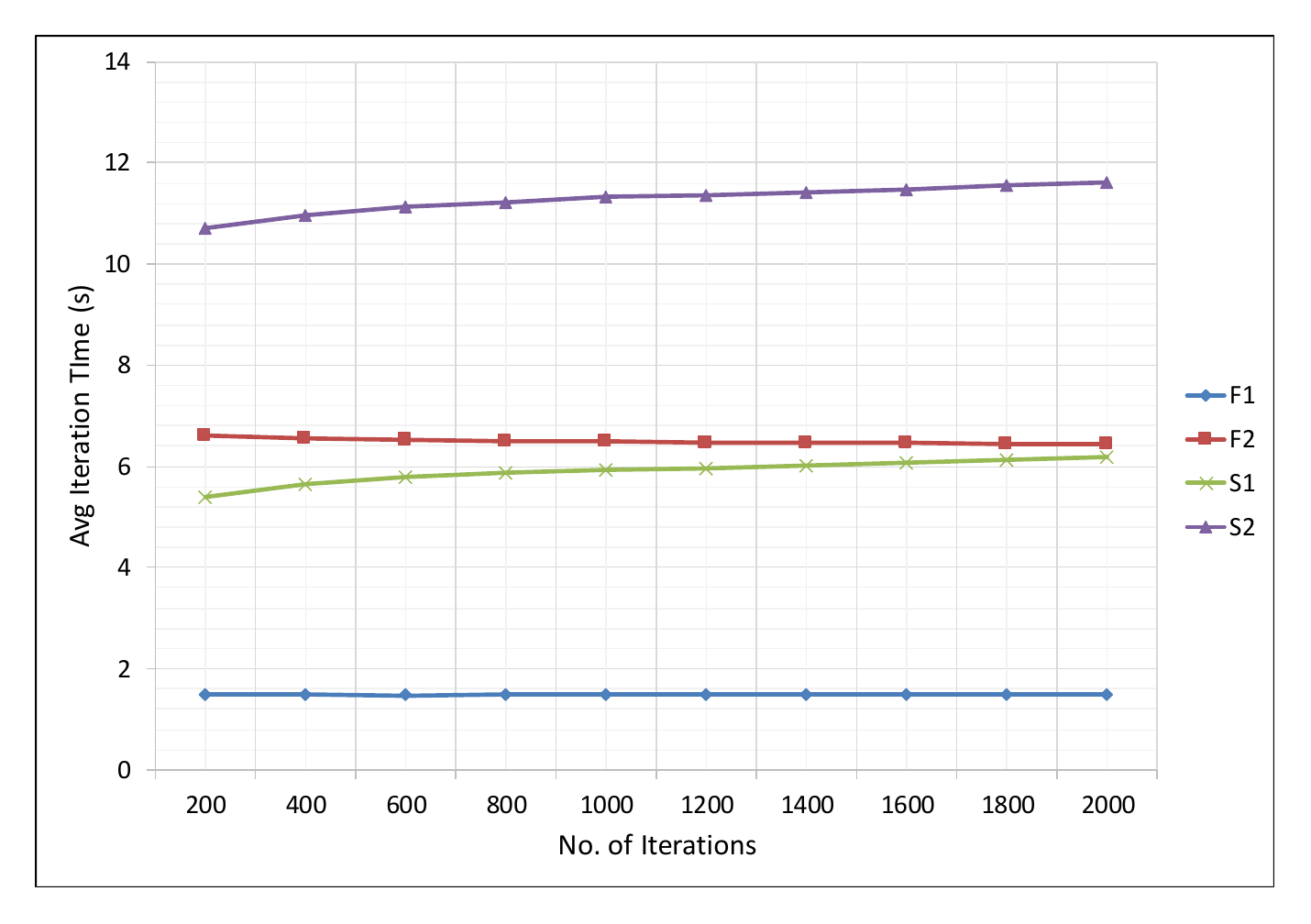}
\caption{Iteration times of accelerated RLT2-DA variants.}
\label{fig:ch4:nug20times}
\end{center}
\end{figure}

\begin{itemize}
\item In general, for non-SA variants, we can see that the lower bounds obtained using the ``slow'' variants are stronger than the ones obtained using the ``fast'' variants. The lower bounds obtained using the ``2-phase'' variants are stronger than the ones obtained using the ``1-phase'' variants.
\item  The lower bounds obtained using SA are significantly stronger than their non-SA counterparts. However, the lower bounds obtained using the ``fast'' variants with SA are much stronger than the ones obtained using the ``slow' variants with SA, contrary to Equation \eqref{eq:ch4:proof:4}. The reason behind this behavior is not completely understood but we suspect that the ``fast'' variants allow for a better redistribution of the LB among the cost coefficients of the symmetrical $z$ variables.
\item  The iteration times of ``fast'' variants are significantly shorter than the ``slow'' variants, and we can see that ``1-phase'' variants are at least twice as fast as ``2-phase'' variants. Additionally, the average iteration time for the  ``fast'' variants remains more or less constant. However, for the ``slow'' variants, the average iteration time increases with the number of iterations. The reason for this phenomenon is that as we update the dual multipliers, cost coefficients are spread further apart, thereby increasing the time spent in ``augmenting path search'' and ``dual update'' steps of the Hungarian algorithm \citep{date2016gpu}.
\item In general, S2-RLT2-DA is dominated both in terms of execution time and lower bound strength.
\item The primary bottleneck in the iteration time RLT2-DA is the Z-TLAP solution phase, however, we can increase the number of PEs (up to a certain limit) and solve more TLAPs in parallel to further reduce the iteration time. We present this scalability study in the next section.

\end{itemize}

\subsection{Multi-GPU Scalability Study}

Although, there is a minimum required number of PEs for applying accelerated RLT2-DA to a QAP of specific size, the number of PEs can be increased and the LAPs can be solved parallely on multiple PEs. This allows us to achieve some parallel speedup. We performed strong scalability study of our accelerated S1-RLT2-DA algorithm  on four of the Nugent problem instances. For this study, the PEs were increased from 1 to 32 in geometric increments of 2. For each PE category, the Y-TLAPs and Z-TLAPs were split evenly across all the GPUs and  solved parallely during each iteration. The results for the parallel scalability study are shown in Table \ref{tbl:ch4:scalability} and Fig.\ \ref{fig:ch4:scalability}. We can see that, as we continue to increase the number of PEs in the system, we get diminishing returns in the execution times. In other words, doubling the number of PEs does not necessarily reduce the execution time by half. This happens because increasing the number of PEs also increases the MPI communication. At some point, adding more PEs in the system will actually increase the execution time, because the communication will start to dominate.

\begin{table}
\begin{center}
\caption{Strong scalability results for S1-RLT2-DA (200 iterations).}
\label{tbl:ch4:scalability}
{\begin{tabular}{cccccccc}
\hline
Problem & LAP Counts & \multicolumn{6}{c}{Iteration Time (s)}\tabularnewline
\cline{3-8}
Instance & (X, Y, Z) & 1 PE & 2 PE & 4 PE & 8 PE & 16 PE & 32 PE\tabularnewline
\hline
Nug20 & (1, 400, 72200) & 5.41 & 3.36 & 2.68 & 2.40 & 1.36 & 1.88\tabularnewline

Nug22 & (1, 484, 106722) & 8.35 & 5.56 & 3.31 & 2.94 & 2.55 & 2.30\tabularnewline

Nug25 & (1, 625, 180000) & 15.94 & 10.11 & 7.26 & 5.08 & 3.68 & 3.49\tabularnewline

Nug27 & (1, 729, 246402) & 27.33 & 15.65 & 10.42 & 7.17 & 6.12 & 6.06\tabularnewline

\hline
\end{tabular}
}
\end{center} 
\end{table}

\begin{figure}
\begin{center} 
\includegraphics[width=0.7\columnwidth,page=12]{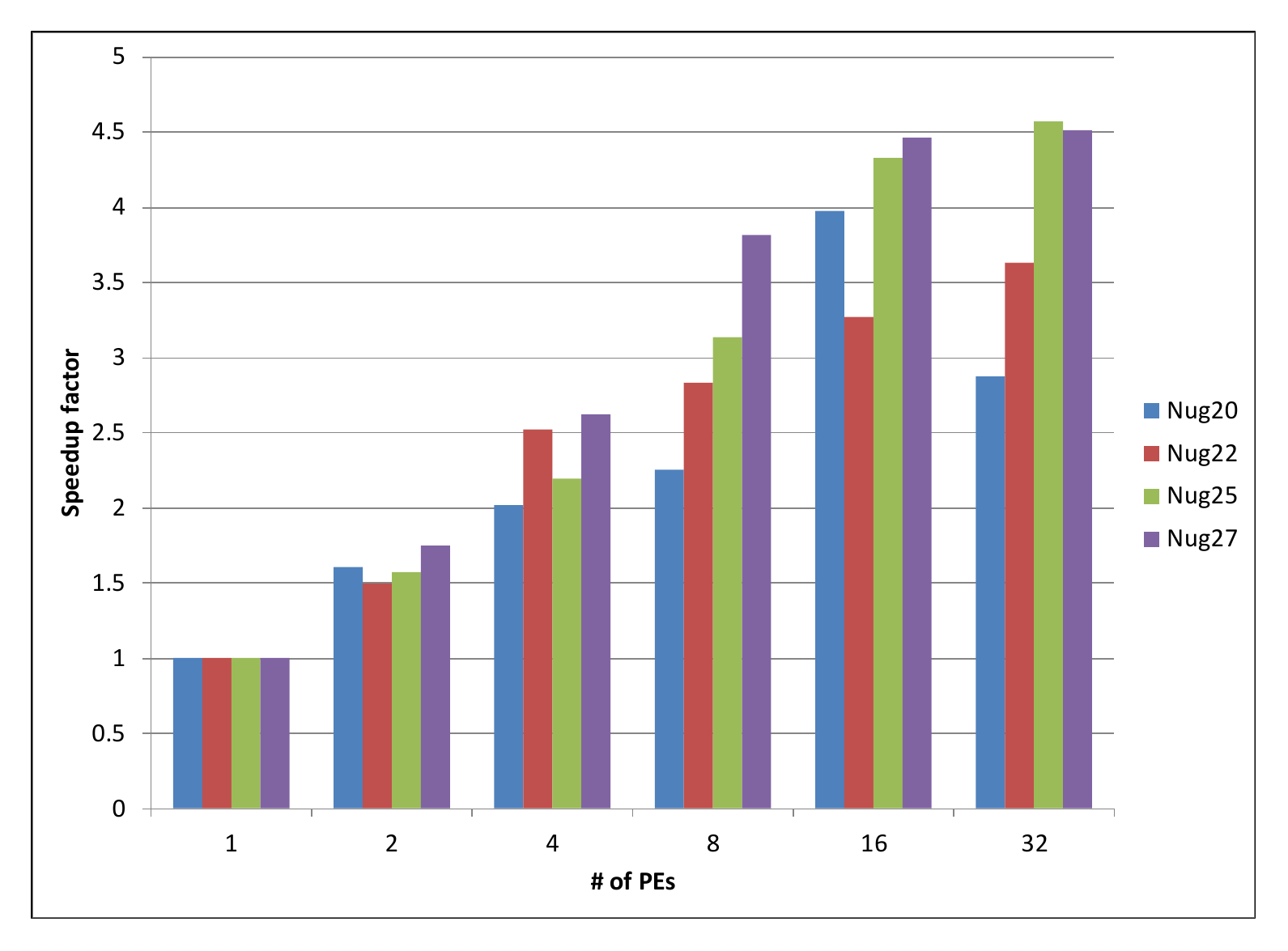}
\caption{Strong scalability tests for S2-RLT2-DA on Nug20.}
\label{fig:ch4:scalability}
\end{center} 
\end{figure}

\subsection{Lower Bounds on QAPLIB Instances}

We also tested the SA-based F2-RLT2-DA on some of the other well-known instances from the QAPLIB; and compared the lower bounds with three other lower bounding techniques, namely Triangle Decomposition Bound (TDB) of \citet{karisch1995}, Lift and Project relaxation bound (L\&P) of \citet{burer2006}, and  Semi-definite Programming (SDP) relaxation bound of \citet{peng2015}. The results are listed in Tables \ref{tbl:ch4:lbs} and \ref{tbl:sdp}. The percentage optimality gap is calculated as $100\times \frac{\text{UB}- \text{LB}}{\text{UB}}$. We have not reported the execution times of the competing methods, because of the differences in the hardware used by the authors.  

\begin{table}
\begin{center} 
\caption{F2-RLT2-DA lower bounds for various instances from QAPLIB (2000 iterations).}
\label{tbl:ch4:lbs}
{%
\begin{tabular}{ccccccc}
\hline
Problem & LAP Counts (X, Y, Z) & $K$ & UB & LB & \% GAP & Itn Time (s)\tabularnewline
\hline
Nug20 & (1, 400, 72200) & 1 & $2570^*$ & $2520$ & 1.95 & 6.52\tabularnewline
Nug22 & (1, 484, 106722) & 2 & $3596^*$ & $3557$ & 1.08 & 6.48\tabularnewline
Nug25 & (1, 625, 180000) & 4 & $3744^*$ & $3610$ & 3.58 & 7.99\tabularnewline
Nug27 & (1, 729, 246402) & 7 & $5234^*$ & $5076$ & 3.02 & 8.49\tabularnewline
Nug30 & (1, 900, 378450) & 12 & $6124^*$ & $5846$ & 4.54 & 9.08\tabularnewline
\hline
Tai20a & (1, 400, 72200) & 1 & $703,482^*$ & $675,257$ & 4.01 & 6.45\tabularnewline
Tai20b$^\ddagger$ & (1, 400, 72200) & 1 & $122,455,319^*$ & $122,454,000$ & 0.00 & 10.33\tabularnewline
Tai25a & (1, 625, 180000) & 4 & $1,167,256^*$ & $1,091,480$ & 6.49 & 7.75\tabularnewline
Tai25b$^\ddagger$ & (1, 625, 180000) & 4 & $344,355,646^*$ & $344,353,000$ & 0.00 & 9.89\tabularnewline
Tai30a & (1, 900, 378450) & 12 & $1,818,146$ & $1,687,500$ & 7.19 & 8.98\tabularnewline
Tai30b & (1, 900, 378450) & 12 & $637,117,113^*$ & $620,444,000$ & 2.62 & 10.73\tabularnewline
Tai35a & (1, 1225, 708050) & 31 & $2,422,002$ & $2,178,360$ & 10.06 & 14.57\tabularnewline
Tai35b & (1, 1225, 708050) & 31 & $283,315,445$ & $266,914,000$ & 5.79 & 17.44\tabularnewline
Tai40a & (1, 1600, 1216800) & 71 & $3,139,370$ & $2,777,580$ & 11.52 & 17.79\tabularnewline
Tai40b & (1, 1600, 1216800) & 71 & $637,250,948$ & $600,983,000$ & 5.69 & 19.63\tabularnewline
\hline
Tho30 & (1, 900, 378450) & 12 & $149,936^*$ & $140,827$ & 6.08 & 9.22\tabularnewline
Tho40 & (1, 1600, 1216800) & 71 & $240,516$ & $213,372$ & 11.29 & 18.41\tabularnewline
\hline
Sko42 & (1, 1764, 1482642) & 95 & $15,812$ & $14,741$ & 6.77 & 21.54\tabularnewline
\hline
\multicolumn{7}{l}{\footnotesize $^\ddagger$ Gap closure was achieved for these problems in 110 and 684 iterations respectively.}\tabularnewline
\end{tabular}
}
\end{center} 
\end{table}

\begin{table}
\begin{center} 
\caption{Comparison lower bounds from other methods with RLT2-DA.}
\label{tbl:sdp}
{%
\begin{tabular}{cccccc}
\hline
\multirow{2}{*}{Problem} & \multirow{2}{*}{UB} & \multicolumn{4}{c}{\% GAP}\tabularnewline
\cline{3-6}
 &  & RLT2-DA & TDB$^1$ & L\&P$^2$ & SDRMS-SUM$^3$\tabularnewline
\cline{1-6}
Nug20 & $2570^*$ & 1.95 & 6.85 & 2.49 & 9.03\tabularnewline
Nug22 & $3596^*$ & 1.08 & - & 2.34 & 8.68\tabularnewline
Nug25 & $3744^*$ & 3.58 & - & 3.29 & 9.05\tabularnewline
Nug27 & $5234^*$ & 3.02 & - & 2.33 & 7.91\tabularnewline
Nug30 & $6124^*$ & 4.54 & 5.75 & 3.10 & 8.43\tabularnewline
\hline
Tai20a & $703,482^*$ & 4.01 & - & 4.52 & -\tabularnewline
Tai20b & $122,455,319^*$ & 0.00 & - & 4.11 & 6.39\tabularnewline
Tai25a & $1,167,256^*$ & 6.49 & - & 4.66 & -\tabularnewline
Tai25b & $344,355,646^*$ & 0.00 & - & 11.70 & 15.15\tabularnewline
Tai30a & $1,818,146$ & 7.19 & - & 6.12 & -\tabularnewline
Tai30b & $637,117,113^*$ & 2.62 & - & 18.47 & 12.72\tabularnewline
Tai35a & $2,422,002$ & 10.06 & - & 8.48 & -\tabularnewline
Tai35b & $283,315,445$ & 5.79 & - & 15.42 & 13.78\tabularnewline
Tai40a & $3,139,370$ & 11.52 & - & 9.43$^\dagger$ & \tabularnewline
Tai40b & $637,250,948$ & 5.69 & - & - & 11.13\tabularnewline
\hline
Tho30 & $149,936^*$ & 6.08 & 9.00 & 4.75 & 12.24\tabularnewline
Tho40 & $240,516$ & 11.29 & 10.94 & 6.69$^\dagger$ & -\tabularnewline
\hline
Sko42 & $15,812$ & 6.77 & 5.56 & - & 7.58\tabularnewline
\hline
\multicolumn{6}{l}{\footnotesize $^1$\citet{karisch1995}. $^2$\citet{burer2006}. $^3$\citet{peng2015}.}\tabularnewline
\multicolumn{6}{l}{\footnotesize $^\dagger$These bounds were reported in QAPLIB and not in the referenced article.}\tabularnewline
\end{tabular}
}
\end{center} 
\end{table}

We can see that RLT2-DA dominates the TDB and SDRMS-SUM methods in almost all of the test instances, while it is dominated by the L\&P method on instances with $n > 22$. 
However, on TaiXXb instances, RLT2-DA performs extremely well, while L\&P has unexpectedly poor performance. We suspect this happens because TaiXXb instances have asymmetric distances, that are weakly planar and widely separated, which somehow favors the RLT2 polytope and the corresponding lower bound. The only downside of using RLT2-DA is that it is time-intensive and requires large memory for modest-sized instances  (see Table \ref{tbl:ch4:mingpu}). Other methods can provide strong lower bounds, possibly faster than RLT2-DA, and they do not require large amounts of memory, even for sufficiently large QAP instances. A hybrid method could be proposed as future work.

\subsection{Computational Results for Parallel Branch-and-bound}

Finally, we tested our parallel B\&B scheme with accelerated RLT2-DA. The master list was seeded with nodes from $\ell_{\text{init}} \in \{2, 3, 4\}$, i.e., nodes obtained by fixing the locations of the specified number of facilities, using Branching Rule 3. For each of the nodes allocated to a particular PE bank, a lower bound is calculated using at most 500 iterations of our accelerated F1-RLT2-DA with SA.  The RLT2-DA procedure was stopped, in favor of branching, if the optimality gap did not improve by $0.0002$ within the last $25$ iterations. The maximum branching depth $d$ was set to 5. For each instance, the best-known upper bound was used as the initial bound estimate in the B\&B scheme.
For each problem, we noted the total execution time, the number of nodes explored, and the utilization factor of each PE bank. The utilization factor is equal to the ratio of clock time for which a particular PE bank was busy with productive work, such as processing a node, to the total clock time for which the resources were requested. Low utilization indicates that the PE bank spent most of its time in idle state. 

\begin{table}[!h]
\begin{center}
\caption{Branch-and-bound results on QAPLIB instances.}
\label{tbl:ch4:bnb}
{%
\begin{tabular}{ccccccccccc}
\hline
\multirow{2}{*}{Problem} & \multirow{2}{*}{UB}  & \multirow{2}{*}{$N$} & \multirow{2}{*}{$K$} & \multirow{2}{*}{$\ell_{\text{init}}$} & \multicolumn{2}{c}{Nodes} & \multicolumn{3}{c}{PE Utilization} & Time\tabularnewline
\cline{6-10}
 &  &  &  &  & Initial  & Explored   & {Min} & {Avg} & {Max} & (d:hh:mm:ss) \tabularnewline
\cline{1-11}

Nug20$^\dagger$ & $2570^*$ & 4 & 1 & 2 & 98 & 134 & 0.77 & 0.89 & 0.99 & 0:00:38:41\tabularnewline
Nug22 & $3596^*$ & 10 & 1 & 2 & 462 & 622 & 0.82 & 0.91 & 0.99 & 0:01:34:03\tabularnewline
Nug25$^\dagger$ & $3744^*$ & 50 & 2 & 3 & 1,755 & 3,868 & 0.81 & 0.90 & 0.97 & 0:02:44:24\tabularnewline
Nug27 & $5234^*$ & 100 & 2 & 3 & 17,550 & 55,761 & 0.97 & 0.98 & 0.99 & 1:02:28:32\tabularnewline
Nug30$^\dagger$ & $6124^*$ & 300 & 4 & 4 & 164,520 & 840,273 & 0.96 & 0.96 & 0.97 & 4:14:06:21\tabularnewline
\hline
Tai20a & $703,482^*$ & 10 & 1 & 2 & 380 & 3,512 & 0.88 & 0.92 & 0.98 & 0:03:56:52\tabularnewline
Tai20b$^\ddagger$ & $122,455,319^*$ & 1 & 1 & 0 & 1 & 1 & 1.00 & 1.00 & 1.00 & 0:00:18:57\tabularnewline
Tai25a & $1,167,256^*$ & 100 & 2 & 3 & 13,800 & 523,005 & 0.97 & 0.98 & 0.98 & 3:13:53:33\tabularnewline
Tai25b$^\ddagger$ & $344,355,646^*$ & 1 & 4 & 0 & 1 & 1 & 1.00 & 1.00 & 1.00 & 0:01:52:43\tabularnewline
Tai30b & $637,117,113^*$ & 60 & 4 & 2 & 870 & 30,523 & 0.91 & 0.93 & 0.95 & 2:09:55:17\tabularnewline

\hline

\multicolumn{11}{l}{\footnotesize $^\dagger$ Grid symmetry elimination rules were used for these problem instances.}\tabularnewline
\multicolumn{11}{l}{\footnotesize $^\ddagger$ Gap closure was achieved for these problems in 110 and 684 iterations respectively.}
\end{tabular}
}
\end{center} 
\end{table}

The results for the B\&B tests are shown in Table \ref{tbl:ch4:bnb}. We can see that the number of nodes explored and the completion times increase exponentially with the problem size. The PE bank utilization also increases, because there is more work available for each processor. The number of nodes explored in each problem are comparable to those reported by \citet{adams2007}. The most notorious Nug30 problem instance required over 4 days to solve optimally, using 300 PE banks with 4 PEs each. The number of nodes explored were 840K. To put this in perspective, the  solution procedure proposed by \citet{anstreicher2002} used an average of 650 worker machines (peak 1000) for over a one-week period. The number of nodes explored were of the order of 11G. \citet{anstreicher2002} also reported that, due to the unstructured flow and distance matrices of the TaiXXa problems, solving those problems using their MWQAP procedure is not practical, even with large computational grids.  
Although, we noted the weakening of the RLT2 lower bounds for TaiXXa instances, our parallel algorithm designed for GPU cluster, was able to solve the Tai20a instance using 10 PEs in 4 hours; the Tai25a instance using 200 PEs in 3.58 days; and the Tai30b instance using 240 PEs in 2.42 days.  Since our parallelization scheme is scalable across multiple GPUs, these solution times could have been further improved 
by simply requesting more GPUs. Our tests were limited by the number of hours allocated to our project on the Blue Waters system.



\section{Conclusions}\label{sec:qap:concl}

To summarize, we developed a Graphics Processing Units (GPU)-accelerated version of the Lagrangian dual ascent procedure (RLT2-DA), for obtaining lower bounds on the Level-2 Refactorization Linearization Technique (RLT2)-based formulation of the Quadratic Assignment Problem (QAP), using multiple GPUs in a grid setting. The sequential procedure has two main stages: Linear Assignment Problem (LAP) solution stage and multiplier update stage. In the LAP solution stage, we have to solve $O(n^6)$ LAPs of size $n \times n$, which can be solved independently of each other. We can use our GPU-accelerated Hungarian algorithm to solve a group of LAPs on each GPU, which provides additional speed up. For the multiplier update stage, we leveraged on the fact that each multiplier can be updated independently of the others, and this can be done parallely by a host of CUDA threads. Our main contribution is a novel GPU-based parallelization of the RLT2-DA, in which we used redundant matrices for the symmetrical cost coefficients. This approach allows us to avoid communicating $O(n^6)$ cost coefficients through MPI, and  achieve superior parallel scalability. 

We conducted several tests on different variants of our accelerated RLT2-DA procedure, and compared them based on their lower bound strengths and execution times. 
We concluded that simulated annealing based approaches provide significantly stronger bounds as compared to non-simulated annealing based approaches. Although it is counter-intuitive, simulated annealing based fast 2-phase (F2) variant provides the strongest lower bound of all the variants. Therefore, it is best suited for settings where the primary goal is to find strong lower bounds on the QAPs.  The fast 1-phase (F1) and slow 1-phase (S1) variants  play an important role in branch-and-bound, where we need to consider the trade-off between the bound strength and iteration time. 
With our architecture, we are able to obtain strong lower bounds on problem instances with up to 42 facilities and optimally solve problems with up to 30 facilities, using only a modest number of GPUs.  

The most lucrative feature of our GPU-accelerated algorithm is that it can be implemented on a desktop computer simply equipped with a few gaming graphics cards, to obtain strong lower bounds on comparatively large QAPs. With some additional work in memory management and CPU + GPU collaboration, our proposed algorithms can be used effectively to solve truly large QAPs with $n\geq 30$, by harnessing the full potential of supercomputing systems like the Blue Waters. Other future directions of research include adaptation of RLT2-DA for solving related problems such as the facility location, graph association, traveling salesman problem, vehicle routing problem, etc. 


\section*{Acknowledgments}

Development and testing of this work was done using the resources from the Blue Waters sustained-petascale computing project, which is supported by the National Science Foundation (awards OCI-0725070 and ACI-1238993) and the state of Illinois. Blue Waters is a joint effort of the University of Illinois at Urbana-Champaign and its National Center for Supercomputing Applications. We gratefully appreciate this support.

\bibliography{bib_qap,bib_lap}
\bibliographystyle{apalike} 

\end{document}